\PassOptionsToPackage{usenames,dvipsnames}{xcolor}

\documentclass[a4paper,USenglish,cleveref, autoref, thm-restate]{lipics-v2021}

\title{Conflict and Fairness in Resource Allocation}

\author{Susobhan Bandopadhyay}{National Institute of Science, Education and Research, An OCC of Homi Bhabha National Institute, Bhubaneswar 752050, Odisha, India.}{susobhan.bandopadhyay@niser.ac.in}{}{}

\author{Aritra Banik}{National Institute of Science, Education and Research, An OCC of Homi Bhabha National Institute, Bhubaneswar 752050, Odisha, India.}{aritra@niser.ac.in}{}{}

\author{Sushmita Gupta}{The Institute of Mathematical Sciences, HBNI, Chennai, India}{sushmitagupta@imsc.res.in}{}{}

\author{Pallavi Jain}{Indian Institute of Technology Jodhpur, India}{pallavi@iitj.ac.in}{}{}

\author{Abhishek Sahu}{National Institute of Science, Education and Research, An OCC of Homi Bhabha National Institute, Bhubaneswar 752050, Odisha, India.}{abhiksheksahu@niser.ac.in}{}{}

\author{Saket Saurabh}{The Institute of Mathematical Sciences, HBNI, Chennai, India}{saket@imsc.res.in}{}{}

\author{Prafullkumar Tale}{Indian Institute of Science Education and Research Bhopal, Bhopal, India \and \url{https://pptale.github.io/}}{prafullkumar@iiserb.ac.in}{}{}

\authorrunning{Bandopadhyay et al.}

\Copyright{Bandopadhyay et al.}



\relatedversion{An earlier version of the paper~\cite{DBLP:journals/corr/abs-2309-04995}, with a subset of authors, contains some of the results mentioned in this article.} 

\hideLIPIcs  

\nolinenumbers 

\usepackage{amsmath}

\usepackage{mathrsfs} 
\usepackage{xspace}
\usepackage{diagbox}
\usepackage{amsfonts}
\usepackage{enumitem}

\usepackage{enumitem}
\usepackage{nicefrac}
\usepackage{comment}

\usepackage{tikz}
\usetikzlibrary{shapes,arrows}
\usetikzlibrary{positioning}

\usepackage[textsize=small]{todonotes}

\usepackage{mathtools}

\usepackage{tabularx}
\usepackage{multirow}
\usepackage{booktabs}
\usepackage{array}
\usepackage{tcolorbox}

\usepackage{rotating}

\newcommand{\defprob}[3]{
\begin{tcolorbox}[colback=gray!5!white,colframe=gray!75!black]
  \begin{tabular*}{\textwidth}{@{\extracolsep{\fill}}lr} #1   \\ \end{tabular*}
  {\bf{Input:}} #2  \\
  {\bf{Question:}} #3
  \end{tcolorbox}
}

\usepackage{complexity}

\newcommand{\defproblem}[3]{
  \vspace{1mm}
\noindent\fbox{
  \begin{minipage}{.95\textwidth}
  \begin{tabular*}{\textwidth}{@{\extracolsep{\fill}}lr} #1 \\ \end{tabular*}
  {\bf{Input:}} #2  \\
  {\bf{Question:}} #3
  \end{minipage}
  }
  \vspace{1mm}
}

\newcommand{\cffafull}{{\sc Conflict free Fair Allocation}\xspace}
\newcommand{\cffa}{{\sc CFFA}\xspace}
\newcommand{\pcffa}{{\sc P-CFFA}\xspace}
\newcommand{\ccffa}{{\sc C-CFFA}\xspace}
\newcommand{\sbccffa}{{\sc Sb-C-CFFA}\xspace}
\newcommand{\sbpcffa}{{\sc Sb-P-CFFA}\xspace}
\newcommand{\cffag}{{\sc ${\cal{G}}$-Sb-P-CFFA}\xspace}

\newcommand{\dcffa}{{\sc Sb-CFFA}\xspace}

\newcommand{\ptime}{|{\mathscr J}|^{\Oh(1)}\xspace}

\newcommand{\util}{\ensuremath{\textsf{ut}}}

\newcommand{\mwis}{{\sc Sb-MWIS}\xspace}

\newcommand{\sbmwisfull}{{\sc Size bounded Max Weight Independent Set}\xspace}
\newcommand{\sbmwis}{{\sc Sb-MWIS}\xspace}

\newcommand{\partialfairdiv}{\textsc{P-CFFA}\xspace}

\newcommand{\completefairdiv}{\textsc{C-CFFA}\xspace}

\newcommand{\mwisg}{{\sc ${\mathcal G}$-\sbmwis}\xspace}

\newcommand{\nph}{{\sf NP}-hard\xspace}
\newcommand{\npc}{{\sf NP}-complete\xspace}

\newcommand{\ETH}{{\sf ETH}\xspace}

\newcommand{\fpt}{{\sf FPT}\xspace}

\newcommand{\para}{{\sf para}\xspace}

\newcommand{\Oh}{\ensuremath{\mathcal{O}}\xspace}

\newcommand{\calA}{\ensuremath{\mathcal{A}}\xspace}
\newcommand{\calB}{\ensuremath{\mathcal{B}}\xspace}

\newcommand{\calO}{\ensuremath{\mathcal{O}}\xspace}
\newcommand{\calH}{\ensuremath{\mathcal{H}}\xspace}

\newcommand{\calI}{\ensuremath{\mathcal{I}}\xspace}

\newcommand{\Co}[1]{\ensuremath{\mathcal{#1}}\xspace}
\newcommand{\X}[1]{\ensuremath{\mathscr{#1}}\xspace}

\newcommand{\conflict}{conflict\xspace}
\newcommand{\object}{{job}\xspace}
\newcommand{\objects}{{jobs}\xspace}
\newcommand{\items}{{jobs}\xspace}

\newtheorem{defn}{Definition}

\newcommand{\yes}{yes}
\newcommand{\no}{no}

\newcommand{\sse}{\subseteq}

\usepackage{thmtools, thm-restate} 
\newtheorem*{def-nonr}{Definition}
\newtheorem{remk}{Remark}

\newcommand{\il}[1]{\todo[inline, color=yellow!70]{#1}}

\newcommand{\susobhan}[1]{\todo[inline, color=Blue!40]{Susobhan: #1}}

\newcommand{\Ma}[1]{\textcolor{magenta}{#1}}

\newcommand{\Rw}[1]{\textcolor{orange}{#1}}

\newcommand{\hide}[1]{}

\usepackage{comment}
\excludecomment{dontshow}

\usepackage{cleveref}

\usepackage{multicol}
\usepackage{array}

\usepackage{colortbl}
\usepackage{tabulary}
\newcolumntype{K}[1]{>{\centering\arraybackslash}p{#1}}
\usepackage{xcolor}
\usepackage{mdframed}
\usepackage{subcaption}

\newcommand{\threedmatch}{$3$-\textsc{Dimensional Matching}\xspace}
\newcommand{\kcol}{$k$-\textsc{coloring}\xspace}
\newcommand{\threecol}{$3$-\textsc{coloring}\xspace}

\newcommand{\tw}{\textsf{tw}\xspace}

\newcommand{\vc}{\textsf{vc}\xspace}

\newcommand{\Sc}{\mathcal{S}\xspace}
\newcommand{\calG}{\mathcal{G}\xspace}

\newcommand{\pnph}{Para-\textsf{NP}-hard\xspace}

\newcommand{\pallavi}[1]{{\color{blue}{Pallavi says: #1}}}

\newcommand{\I}{\ensuremath{\mathcal{I}}\xspace}

\EventEditors{John Q. Open and Joan R. Access}
\EventNoEds{2}
\EventLongTitle{42nd Conference on Very Important Topics (CVIT 2016)}
\EventShortTitle{CVIT 2016}
\EventAcronym{CVIT}
\EventYear{2016}
\EventDate{December 24--27, 2016}
\EventLocation{Little Whinging, United Kingdom}
\EventLogo{}
\SeriesVolume{42}
\ArticleNo{23}

\begin{document}



\maketitle 


\begin{abstract}
In the standard model of fair allocation of resources to agents, 
every agent has some utility for every resource, and the goal is to assign 
resources to agents so that the agents' welfare is maximized.
Motivated by job scheduling, interest in this problem dates back to the 
work of Deuermeyer et al. [SIAM J. on Algebraic Discrete Methods'82]. Recent 
works consider the compatibility between resources and assign only mutually 
compatible resources to an agent.
We study a fair allocation 
problem in which we are given a set of agents, a set of resources, a utility 
function for every agent over a set of resources, and a {\it conflict graph} on 
the set of resources (where an edge denotes incompatibility).
The goal is to 
assign resources to the agents such that 
$(i)$ the set of resources allocated to an agent are compatible with each 
other, and 
$(ii)$ the minimum satisfaction of an agent is maximized, where the 
satisfaction of an agent is the sum of the utility of the assigned resources. 
Chiarelli et al. [Algorithmica'22] explore this problem from the classical 
complexity perspective to draw the boundary between
the cases that are polynomial-time solvable and those that are 
\NP-hard. 
In this article, we study the parameterized complexity of the problem
(and its variants) by considering several 
natural and structural parameters.
\end{abstract}

\section{Introduction}

Resource allocation is a central topic in economics and computation. 
It is an umbrella term that captures a plethora of well-known problem 
settings where resources are matched to agents in a meaningful way that 
respects the preferences/choices of agents, and when relevant, resources as 
well.
Stable matching, generalized assignment, and fair division are well-known problems that fall under the purview of resource allocation. 
These topics are extensively studied in economics, (computational) social 
choice theory, game theory, and computer science, to name a few, and are 
incredibly versatile and adaptable to a wide variety of terminology, 
techniques and traditions. 

A well-known framework within which resource allocation is studied 
is the world of {\sc Job Scheduling} problems on non-identical machines.
In this scenario, the machines are acting as agents, and the jobs are the tasks such that specific jobs are better suited for a given machine  than others and 
this variation is captured by the ``satisfaction level'' of the machine towards 
the assigned jobs.
Moreover, the jobs have specific time intervals within 
which they have to be performed, and only one job can be scheduled on a 
machine at a time.
Thus, the subset of jobs assigned to a single machine 
must respect these constraints, and the objective can be both maximization 
and minimization, as well as to simply test feasibility. Results on the 
computational aspect of resource allocation that incorporates interactions 
and dependencies between the resources are relatively few. A rather 
inexhaustive but representative list of papers that take a combinatorial 
approach in analyzing a resource allocation problem and are aligned with our 
work in this paper is 
\cite{cheng2022restricted,bezakova2005allocating,kurokawa2018fair,DBLP:conf/atal/EbadianP022,DBLP:conf/atal/BarmanV21,ahmadian2021four,DBLP:conf/iwoca/ChiarelliKMPPS20,DBLP:journals/sigecom/Suksompong21,woeginger1997polynomial, hummel2022fair, Biswas_FORC_2023}. 

The starting point of our research is the 
decades-old work of Deuermeyer et al.~\cite{deuermeyer1982scheduling} 
that studies a variant of {\sc Job Scheduling} in which the goal is to assign a 
set of independent jobs to identical machines to maximize the 
minimal completion time of the jobs. Their {\sf NP}-hardness result for two 
machines (i.e., two agents in our setting) is an early work with a similar flavor. 
The more recent work of Chiarelli et. al.~\cite{DBLP:conf/iwoca/ChiarelliKMPPS20} 
that studies ``fair allocation'' of indivisible items into pairwise disjoint 
conflict-free subsets that maximizes the minimum satisfaction of the agents is the work closest to ours. 
They, too, consider various graphs, called the {\it conflict graph}, that capture the conflict/compatibilities between various items and explore the classical complexity boundary between strong 
{\sf NP}-hardness and pseudo-polynomial tractability for a constant number 
of agents. 
Our analysis probes beyond the {\sf NP}-hardness of these problems and 
explores this world from the lens of \emph{parameterized complexity}, which gives a more 
refined tractability landscape for {\sf NP}-hard problems.

The {parameterized complexity} paradigm allows 
for a more refined analysis of the problem’s complexity.
In this setting, we associate each instance $I$ with 
parameters $\ell_1, \ell_2, \dots$, 
and are interested in an algorithm with a running time 
$f(\ell_1, \ell_2, \dots) \cdot |I|^{\calO(1)}$ 
for some computable function $f$.
When these parameters are significantly smaller than 
the total size of the input,  we get a much faster algorithm
than the brute-force algorithm that is expected 
to run in time exponential in the input size for \nph problems.
Parameterized problems that admit such an algorithm are called 
\emph{fixed-parameter tractable} (\fpt) with respect to the parameters under consideration.
On the other hand,
parameterized problems that are hard for the complexity class {\sf W[r]} for any $r\geq 1$ do not admit fixed-parameter algorithms with respect to that parameter, under a standard complexity assumption.
If for a particular choice of parameter $\ell$, 
the problem remains \nph even when $\ell$ is a constant then, 
we say the problem is \para-\nph when parameterized
by $\ell$.
A parameter may originate from the formulation of the problem itself, and is 
called a \emph{natural parameter}, or it can be a property of the conflict graph, and is called a \emph{structural parameter}.

Next, we define our problem and discuss the parameters of interest.
The problem can be viewed as a two-sided matching 
between a set of agents and a set of \items.
Each worker (i.e., an {\it agent}) has a utility function defined 
over the set of {\it \items}.
Their satisfaction for a {\it bundle}--a subset of 
\objects that induces an independent set in \Co{H}-- of \objects is the sum of the agents' utilities for 
each \object in the bundle. 
The incompatibilities among the \items are captured by a graph \Co{H} defined 
on the set of \items such that an edge represents {\it conflict} 
(incompatibility, in other words). 
The overall objective is to assign pairwise-disjoint bundles to agents such that the minimum satisfaction of the agents is 
maximized. 
A decision version of the above problem can be defined as follows:

\defprob{\cffafull (\cffa)}{A set of $m$ \objects \Co{I} 
and a conflict graph \Co{H} with vertex set \Co{I}, 
a positive integer $\eta$ called target value,
a set of $n$ agents $\Co{A}=\{a_1,\ldots,a_n\}$ and for every agent $a_i\in \Co{A}$, 
a utility function $\util_i\colon \calI \rightarrow \mathbb{Z}_{\geq 0}$.}{Does 
there exist a function  $\phi\colon \Co{I} \rightarrow \Co{A}$  
such that  for every $a_i \in \Co{A}$, $\phi^{-1}(a_i)$ is an independent set in \Co{H} and $\sum_{x\in \phi^{-1}(a_i)}\util_i(x) \geq \eta$?}

We consider the following three variations of the problem viz 
$\textsc{Complete}$, $\textsc{Partial}$, and $\textsc{Size-Bounded}$.
In the first version, we are required to assign every \object to some agent,
i.e., the domain of the function $\phi$ is $\calI$.
The second version relaxes this criterion and 
allows $\phi$ to be a {\em partial function} on $\calI$, i.e, each of the \objects need not be assigned to an agent.
In the third version, an input contains an additional integer $s$
that restricts the number of \objects that can be assigned to
any agent, i.e., we have a further requirement of $|\phi^{-1}(a)| \le s$,
where $\phi$ may be a partial or complete function. {From now onwards, we will use \ccffa for the complete version, \pcffa for the partial version, \sbccffa for the size-bounded complete version, and \sbpcffa for the size-bounded partial versions.}


The natural parameters, in this case, are
the number of \objects $m$, 
the number of agents $n$, 
the threshold value $\eta$, and
the size bound on the bundle $s$. For any non-trivial instance of any version, we have $n \le m$. Otherwise, there will be an agent who does not receive any \object and hence the target value cannot be attained. Moreover, we have $s \le m$ for the size bounded variants. Moreover, it is safe to assume that the range of $\util_i(x)$ is $\{0, 1, \dots, \eta\}$, for all $i\in [n]$.
Note that $\eta$ \emph{need not be} bounded above by a function of $m$. If this is true, we say $\eta$ is \emph{large}. For a class of graphs $\calG$, we define 
$\pi: \calG \mapsto \mathbb{N}$,
where $\pi(\calH)$ denotes some structural properties associated with  graph $\calH$ in $\calG$.
For most of the reasonable structural parameters, it is safe to assume that
$\pi(\calH) \le |V(\calH)| = m$. We set
$t = {m\choose 2} - |E(\Co{H})|$, the number of missing edges in $\calH$. See Figure~\ref{fig:para-relation} for the relations amongstthe parameters considered in this article. 

We start our technical discussion by presenting two results that show that the problems are hard even for very restricted cases.

\begin{figure}[t]
\begin{center}
\tikzstyle{fptnode} = [rectangle, draw, fill=green!20, text width=6em, text centered, rounded corners, minimum height=2em]
\tikzstyle{hardnode} = [rectangle, draw, fill=red!20, text width=6em, text centered, rounded corners, minimum height=2em]
\tikzstyle{line} = [draw, -latex']

\begin{tikzpicture}[auto]
\node[fptnode]     (nrofnodes)   {$m$\\ (Thm~\ref{thm:fpt-nr-of-vertices})};
\node[fptnode]      (missingedgesplusagents) [right=.75cm of nrofnodes] {$t + n$\\ (Thm~\ref{thm:fpt-nr-of-non-edges})};
\node[fptnode]      (sizeplusagent)  [right=.75cm of missingedgesplusagents] {$s + n$\\ (Thm~\ref{thm:equiavlence})};
\node[fptnode]      (agentplusparapluseta) [right= .75cm of sizeplusagent] {$\eta + n + \pi(\calH)$\\ (Thm~\ref{thm:kernel-partial-complete},\ref{thm:kernel-contrast})};


\node[fptnode]   (nrofmissingedges) [below left=2cm and 52pt of missingedgesplusagents] {$t$\\ (Thm~\ref{thm:fpt-nr-of-non-edges})};
\node[hardnode] (sizeofbundle) [right=14pt of  nrofmissingedges]
{$s$\\ (Thm~\ref{thm:np-hardness-results})};
\node[hardnode] (nrofagent) [right=14pt of sizeofbundle]
{$n$\\ (Thm~\ref{thm:np-hardness-results})};
\node[hardnode] (parofgraph) [right=14pt of nrofagent]
{$\pi(\calH)$\\ (Thm~\ref{thm:np-hardness-results})};
\node[hardnode] (threshold) [below right=2cm and 52pt of sizeplusagent]
{$\eta$\\ (Thm~\ref{thm:np-hardness-results})};


\draw[line] (nrofnodes) -- (nrofmissingedges);
\draw[line] (nrofnodes) -- (sizeofbundle);
\draw[line] (nrofnodes) -- (nrofagent);
\draw[line] (nrofnodes) -- (parofgraph);

\draw[line] (missingedgesplusagents) -- (nrofmissingedges);
\draw[line] (missingedgesplusagents) -- (nrofagent);

\draw[line] (sizeplusagent) -- (sizeofbundle);
\draw[line] (sizeplusagent) -- (nrofagent);


\draw[line] (agentplusparapluseta) -- (nrofagent);
\draw[line] (agentplusparapluseta) -- (parofgraph);
\draw[line] (agentplusparapluseta) -- (threshold);
\end{tikzpicture}
\end{center}
\caption{Bird's-eye view of the results in the article.
An arrow from parameter $\ell_1$ to $\ell_2$ implies
that $\ell_2$ is upper bounded by some computable 
function of $\ell_1$.
\hide{Recall that $m$ denotes the number of \object, 
$n$ denotes the number of agents,
$\eta$ denotes the target value, and 
$s$ denotes the size bound on a bundle of each agent. We use
$\pi(\calH)$ to denote a parameter of $\calH$
that is a constant if $\calH$ is an edgeless graph; and
$t = {m\choose 2} - |E(\Co{H})|$, the number of missing edges in $\calH$.}
The red-colored rectangles denote that the problem is \nph 
even when the parameter is a constant.
The green-colored rectangles denote that there is an \FPT\ algorithm. Note that some of these results are conditional or restrictive. 
Please refer to the theorems for precise statements.
\label{fig:para-relation}}
\end{figure}

\begin{restatable}[]{them}{nphardness} \label{thm:np-hardness-results}
We prove that 
\begin{enumerate}
\item \ccffa  
    \begin{enumerate}
    \item admits a polynomial time algorithm if $n = 1$
    even if $\eta$ is large;  and  
    \item is \nph for 
    $(i)$ any $n \ge 2$ even if the binary encoding $\eta$ is a polynomial function
    of $m$; as well as  
    $(ii)$ any $n \ge 3$ even if $\eta = 1$. 
\end{enumerate} 
\item\label{item:IS}
\pcffa is \nph for any $n \ge 1$ even if $\eta$ is at most $m$ (and range of $\util$ is $\{1\}$). 
\item Both \sbccffa\ and \sbpcffa\ 
\begin{enumerate}
     \item admit polynomial-time algorithms when 
     $(i)$ both $n$ and $s$ are constants (even if $\eta$ is large),
      or
    $(ii)$ $s = 1$ (even if $\eta$ is large), and
     \item are \nph\  $(i)$ for any $n \ge 1$ even if the binary encoding of $\eta$ is a polynomial function of $m$, or
    $(ii)$ for any $\eta \ge 2$ and any $s \ge 2$. 
      \end{enumerate}
\end{enumerate}
\end{restatable}

Theorem~\ref{thm:np-hardness-results} 
rules out the possibility of an \FPT\ algorithm when
the parameter under consideration is one of the following:
the number of agents $n$, the maximum size of a bundle $s$,  the threshold
value $\eta$, or a structural parameters $\pi(\calH)$, where 
$\pi(\calH)$ is a constant when $\calH$ is an edgeless graph.
This result, however, does not rule out an \FPT\ algorithm 
when parameterized by the number of items, $m$.
Indeed, note that a simple brute-force algorithm that enumerates all possible 
$(n+1)^{m}$ ways of assigning the \objects to the agents, 
where each \object has $(n+1)$ choices of agents to choose from ( $+ 1$ is required in partial allocation if a \object is unassigned)
runs in time $(n+1)^{m} \cdot (m + \log(\eta))^{\calO(1)}$.
As $n \le m$, this is an \FPT\ algorithm parameterized by $m$.
Our following result significantly improves the running time.


\begin{restatable}[]{them}{fptnrofvertices}
\label{thm:fpt-nr-of-vertices}
All four variants of \cffa admit an algorithm running in time
$2^{m}\cdot (m + \log(\eta))^{\calO(1)}$.
Moreover, unless the \ETH\ fails, 
none of the variations admits an algorithm running in time
$2^{o(m)}\cdot (m + \log(\eta))^{\calO(1)}$.
\end{restatable}

In the following results, we consider a combination of other parameters
mentioned above to obtain some tractable results.
We start with the combined parameter $n + s$. 
{Note that for the complete version,  $m\leq ns$, otherwise, 
it is a trivial no-instance. Thus, Theorem~\ref{thm:fpt-nr-of-vertices} 
implies an \fpt algorithm with respect to $n+s$ for the complete version.
Consequently, for this parameter, we only focus on the partial variants. Note that an \fpt algorithm for \sbpcffa would imply an \fpt algorithm 
for \pcffa, because $s \leq m$ and an algorithm that makes subroutine 
calls to an $\fpt(n+s)$ algorithm for \sbpcffa for increasing values 
of $s$ (from $1$ to $m$) is an $\fpt(n+s)$ algorithm for \pcffa. 
Hence, we focus on designing an algorithm for \sbpcffa. 
However, hoping that \pcffa is \fpt parameterized by $n+s$ in general 
graphs is futile because the problem generalizes the 
{\sc Maximum Weight Independent Set} problem,
which we formally establish in the proof of \hide{the above theorem,} {\Cref{thm:np-hardness-results}(\cref{item:IS})}.
Hence, we can only expect to obtain an $\fpt(n+s)$ algorithm for 
special classes of graphs. Consequently, our exploration moves towards identifying graph classes that may admit such an algorithm.} 
Towards this, we introduce the following problem, 
which we believe could be of independent interest. 
Note that here, the size of the solution (i.e. independent set)  is upper bound by the parameter, distinguishing it from the classic {\sc Maximum Weight Independent Set} problem.

\defproblem{\sbmwisfull (\sbmwis) }{A  graph $G$, positive integers $k$ and 
$\rho$, a weight function $w: V(G) \rightarrow \mathbb{N}$.}{Does there 
exist an independent set $S$ of size at most $k$ such that 
$\sum_{v\in S} w(v)\geq \rho$?}

\sbmwis is a clear generalization of the {\sc Independent Set} problem (by setting $\rho=k$ and unit weight 
function), a very well studied problem in the realm of parameterized complexity and indeed in the wider field of graph algorithms. This connection allows us to demarcate the tractability border of our problem \dcffa via the computational landscape of {\sc Independent Set}. 
In the field of parameterized complexity, {\sc Independent Set}  has been extensively studied on families of graphs that satisfy some structural properties. We take the same exploration path for our problem \sbpcffa. The graph classes in which {\sc Independent Set} has an $\fpt(k)$ algorithm is a potential field for \fpt algorithms for \sbpcffa. This possibility has to be formally explored.

Let $\mathcal G$ be a family of {\em hereditary} graphs. That is, if $G\in {\mathcal G} $, then all the induced subgraphs of $G$ belong to $\mathcal G$. In other words, $\mathcal G$ is closed under taking induced subgraphs. 

%
%
%
 %
%
For a hereditary family \Co{G}, \mwisg denotes the restriction of \mwis where the input graph $G\in \mathcal G$ and \cffag denote the restriction of \sbpcffa where the conflict graph $\Co{H}\in \mathcal G$. \hide{
 }

 The following result completely characterizes the parameterized complexity of \cffag with respect to $n+s$ vis-a-vis \hide{the parameterized complexity of} \mwisg with respect to $k$. 
 

\begin{restatable}{them}{equivalence}\label{thm:equiavlence}
Suppose $\mathcal G$ is a hereditary family of graphs.  
Then, \cffag is \fpt parameterized by $n+s$ if and only if  \mwisg is \fpt parameterized by $k$. 
\end{restatable}

Even though the above theorem establishes the link between
these two problems, it does not identify the graph
classes on which these problems will be tractable.
While there are papers that study hereditary graph classes for the 
{\sc MWIS} problem (see, for example, \cite{Dabrowski12}), 
we are not aware of known classes of graphs for which 
\sbmwis\ is \fpt parameterized by $k$. 
Hence, we introduce the notion of $f$-independence friendly graph classes and prove that the problem admits 
an \fpt\ algorithm when restricted to this graph class.

\begin{def-nonr}[Independence Friendly Graph Class]
Consider $f \colon \mathbb{N} \to \mathbb{R}$ be a 
monotonically increasing and an invertible function.
A graph class $\mathcal G$ is called {\em $f$-independence 
friendly class  ($f$-ifc)} if $\mathcal G$ is hereditary  
and for every $n$-vertex graph $G \in \mathcal G$ has 
an independent set of size $f(n)$. 
\end{def-nonr}

\begin{restatable}{them}{fptfrindlygraph}\label{lem:mwis-degenerate}
Let $\mathcal G$ be an $f$-independence friendly class. Then, 
there exists an algorithm for \mwisg running in time $\Oh((f^{-1}(k))^k\cdot (n+m)^{\Oh(1)})$, and hence an algorithm for \cffag
running in time $2^{\calO(ns)} \cdot (f^{-1}(s))^{s}\cdot (m + \log(\eta))^{\calO(1)}$
\end{restatable}


Some examples of $f$-independence friendly graph classes are 
bipartite graphs, planar graphs, graphs of bounded degeneracy, 
graphs excluding a fixed clique as an induced subgraph, etc.

As the last of our main results, we consider a parameter $\pi$ of 
graph that is not $0$ for edgeless graphs.
One of the most obvious parameters is the number of non-edges.
Note that dense edges, representing multiple conflicts between \items,
naturally bound the number of items that can be assigned to a particular 
agent.
Consider a simple example where $\calH$ is a complete graph on $m$
vertices, i.e., it has $\binom{m}{2}$ many edges.
Then, we can only assign one \items to each agent. 
Hence, the problem reduces to finding a matching in an auxiliary 
bipartite graph across the set of agents and set of \items.
This raises the natural question whether the number of non-edges $t$
is a fruitful parameter.
The following theorem answers this question positively
by presenting an \FPT\ algorithm when parameterized by $t$.
Moreover, if we consider additional parameters as the number 
of agents $n$, we get an improved \FPT\ algorithm.

\begin{restatable}[]{them}{fptnonedges}
\label{thm:fpt-nr-of-non-edges}
All the variants of the \cffa problem,
when restricted to the case where the conflict graph $\calH$ on $m$
vertices has at least $\binom{m}{2} - t$ many edges admits
\begin{enumerate}
\item an algorithm running in time $2^{\calO(n \cdot \sqrt{t} \cdot \log(t) )}
\cdot (m + \log(\eta))^{\calO(1)}$;
and
\item an algorithm running in time $2^{\calO(t \cdot \log(t))} \cdot 
(m + \log(\eta))^{\calO(1)}$.
\end{enumerate}
\end{restatable}


We remark that there are various points of view to consider
the collection of complete graphs and accordingly define 
\emph{the distance from triviality}. 
For example, the collection of graphs that are union
of $q$ many disjoint cliques is in close proximity 
of collection of cliques (for which $q = 1$).
However, our next result shows that such a result
has limited success. 

\begin{restatable}[]{them}{twocliques}
\label{thm:twocliques}
All the variants of \cffa,
when restricted to the case where the conflict graph $\calH$
is a union of $q$ cliques
\begin{enumerate}
\item is solvable in polynomial time when $q = 2$ and utility functions
are uniform, i.e., the utility functions are the same for all the agents; and
\item is \nph when $q = 2$ and $\eta = 3$ (but utility function 
is not uniform). 
\end{enumerate}
\end{restatable}

Finally, we turn our attention to \emph{kernelization} of the problem.
We say that two instances, $(I, k)$ and $(I', k')$, of a parameterized problem $\Pi$ are \emph{equivalent} if $(I, k) \in \Pi$ if and only if $(I', k') \in \Pi$.
A \emph{reduction rule}, for a parameterized problem $\Pi$ is an algorithm that takes an instance $(I, k)$ of $\Pi$ as input and outputs an instance $(I', k')$ of $\Pi$ in time polynomial in $|I|$ and $k$.
If $(I, k)$ and $(I', k')$ are equivalent instances then we say the reduction rule is \emph{safe}.
A parameterized problem $\Pi$ is said to have a {\it kernel} of size $g(k)$ (or $g(k)$-kernel) if there is a polynomial-time algorithm (called a {\em kernelization algorithm}) which takes as an input $(I,k)$, and in time $(|I| + k)^{\calO(1)}$ returns an equivalent instance $(I',k')$ of $\Pi$ such that $|I'| + k' \leq g(k)$.
Here, $g(\cdot)$ is a computable function whose value depends only on $k$.
It is known that a problem $\Pi$ admits an \FPT\ result
if and only if it admits a kernel.
For more details on parameterized complexity,  we refer the reader to the books by Cygan et al.~\cite{ParamAlgorithms15b}.

\begin{restatable}{them}{partialcompletekernel}
    \label{thm:kernel-partial-complete}
    Both \pcffa and \ccffa have the following properties: 
    \begin{enumerate}
        \item have a polynomial-sized kernel of size $\tau \eta n^2$, where $\tau$ denotes the neighborhood diversity of the conflict graph; 
        \item have a polynomial-sized kernel of size $d\eta n^2$ if $d< n$, where $d$ is the maximum degree of the conflict graph; and 
         \item do not have a polynomial-sized kernel when parameterized by ${\sf tw} +n$ unless ${\sf NP}\subseteq {\sf coNP\setminus poly}$, where {\sf tw} be the treewidth of the conflict graph. Moreover, the same result holds even for ${\sf tw} +\eta +n$ on \ccffa.     
    \end{enumerate}\end{restatable}
However, the kernelization complexity starts to differ between \pcffa and \ccffa when in addition to the threshold value we consider structural parameters such as the maximum degree (when greater than the number of agents), chromatic number, and size of the largest clique.

\begin{restatable}{them}{contrastkernel}
    \label{thm:kernel-contrast}~
    \begin{enumerate}
        
        \item If $d>n$, \pcffa has a polynomial-sized kernel of size $d\eta n^2$, where $d$ denotes the maximum degree of the conflict graph.  However, \ccffa becomes \nph even when $d+ \eta +n$ is a constant.
        \item \pcffa admits polynomial-sized kernel of size  $\chi \eta n^2$ 
        and $\calO(n(r+\eta n)^r)$, where $\chi$ is the chromatic number, and $r$ is the size of the maximum clique in the conflict graph. Contrastingly, \ccffa remains \nph for even constant values of $\chi +r+\eta +n$.
    \end{enumerate}\end{restatable}

\section{Proof of Theorem~\ref{thm:np-hardness-results} }

\nphardness*

\begin{proof}~

\begin{enumerate}
\item Results regarding \ccffa.
\end{enumerate}
Consider the case when $n = 1$.
Note that in this case, an instance is a (trivial) \yes-instance if and only
if $\calH$ is an edgeless graph and the sum of the utilities 
of all the items is at least $\eta$.
Both of these checks can be completed in polynomial time.



Consider the case when $n \ge 2$.
We present a reduction from the \textsc{Partition} problem ({Problem~[SP12]} in \cite{DBLP:books/fm/GareyJ79}).
Recall that in this problem, an input is a finite set $P = \{p_1, p_2, \dots, 
p_{|P|}\}$ of integers and the objective is to 
find a subset $P' \subseteq P$ such that 
$\sum_{p_i \in P'} p_i = \sum_{p_{j} \in P \setminus P'} p_j$.
The problem remains \nph\ even when length of encoding of $\max_{p_i \in P} \{p_i\}$ is polynomial in $|P|$.
The reduction constructs an instance of \sbccffa\ by
adding two agents $a_1, a_2$ to $\calA$,
an item $v_i$ to $\calI$ corresponding to every element in set $P$,
and defining $\eta = (\sum_{p_i \in P}p_i)/2$ and
$\util_1(v_i) = \util_2(v_i) = p_i$.
The reduction returns $(\Co{A} = \{a_1, a_2\}, \Co{I}, \{\util_1, \util_2\},
\Co{H}, \eta$) as the reduced instance where $\calH$ is an edgeless graph.
The correctness of the reduction follows easily as the partition of 
vertices corresponds to the assignment of items to two agents.
It is easy to extend the reduction for any value of $n \ge 2$ by 
adding $(n - 2)$ dummy items whose valuation is $\eta$ for any every 
agent.
Finally, the utility functions ensure that all the items
needs to be assigned.

For the case when $n \ge 3$, consider a reduction from the 
\textsc{$3$-Coloring} problem (Problem~[GT4]) in \cite{DBLP:books/fm/GareyJ79}).
In this problem, given graph $G$, the objective is to determine whether
there exists a coloring $\phi: V(G) \mapsto \{1, 2, 3\}$
such that for any edge $uv \in E(G)$, $\phi(u) \neq \phi(v)$.
Consider the reduction that, given an instance of 
the \textsc{$3$-Coloring} problem, construct an instance
of \textsc{ccffa} by adding three agents,
say $a_1, a_2, a_3$, to $\calA$, adding an item with respect
to each vertex in $G$ to $\calI$, and defining $\util_1(v) = \util_2(v) = \util_3(v) = 1$ for all $v \in V(G)$, $\eta = 1$, and $\calH = G$.
It is easy to verify that a proper $3$-coloring $\phi$ of $V(G)$
corresponds to a valid assignment of items to agents
and vice-versa.
A similar reduction starting from $\eta$-$\textsc{Coloring}$
implies that the problem is \nph for any $\eta \ge 3$.



\begin{enumerate}[resume]
\item Results regarding \pcffa.
\end{enumerate}
Consider the \textsc{Independent Set} problem.
The input consists of a graph $G$ and an integer $k$,
and the objective is to find a set of vertices $S$ of size $k$
such that for any two vertices $S$ are not adjacent to each other.
Consider a reduction that, given an instance of \textsc{Independent Set},
constructs an instance of 
\pcffa with one agent 
$a_1 \in \calA$, an item in $\calI$ 
corresponding to every vertex in $V(G)$, $\util_1(v) = 1$ for all $v \in V(G)$,
$\eta = k$, and $\calH = G$.
It is easy to verify that an independent set $S$ of size
$k$ corresponds to a partial assignment of $k$ items 
to the agent.
This implies the desired hardness result. 

\begin{enumerate}[resume]
\item Results regarding \sbccffa\ and \sbpcffa.
\end{enumerate}
We prove that both the problems admit algorithms running in time
$m^{\calO(n s)} \cdot poly(m, \eta)$, 
and hence, it is polynomial-time solvable when $n$ and $s$
are constants.
The algorithm enumerates all the subset of items of size at most
$s$ in time $m^{\calO(s)}$.
These sets will act as a bundle that will be assigned to agents.
It discards the bundles that are not an independent set
in conflict graph $\calH$.
For the $n$ agents, there are $(m^{\calO(s)})^n$ many possible
assignments of these bundles.
The algorithm iterates through all these assignments.
It discards the assignments in which the evaluation of bundle for that 
a particular agent is less than $\eta$ or 
the intersection of two bundles assigned to different agents
is non-empty.
The correctness of the algorithm follows from the fact that
it exhaustively searches all the valid assignments and
the running time follows from the description.

Consider the case when $s = 1$. 
As the valuation of any agent has non-negative integer values,
it suffices to assign one appropriate job to each agent.
Towards this, we consider an auxiliary bipartite graph $\calB$
across agents and jobs and edges representing the feasibility
of useful assignments.
Formally, construct a bipartite graph $\calB$ with a bipartition
$\calA$ and $\calI$. 
For an agent $a_i \in \calA$ and a job $v_j \in \calI$,
we add an edge $(a_i, v_j)$ if and only if $\phi_i(v_j) \le \eta$.
It is easy to see that a matching saturating $\calA$
corresponds to the desired partial assignment of jobs.
For the \sbccffa version, we need a perfect matching, 
whereas for \sbpcffa, we need a matching saturating $\calA$.
As there is a polynomial time algorithm 
to compute the matching~\cite{DBLP:journals/siamcomp/HopcroftK73},  
both problems admit a polynomial-time algorithm when $s = 1$.
The hardness results for the case when $n \ge 1$ follow
from the corresponding case for \pcffa.
Also, the hardness result for the later case follows 
from the corresponding case when $\calH$ is an edgeless graph.

We now consider the case when $\eta \ge 2$ and $s \ge 2$.
We present a reduction from the \threedmatch to the \sbccffa with $\eta=2$
and $s = 2$.
Later, we mention how to extend this result for any $\eta \ge 2$
and $s \ge 2$.
We note that the utility functions and the target
value are defined in such a way that even for an instance
of \sbpcffa, one needs to assign all items to agents.
Hence, the same reduction holds even for \sbpcffa.

An input  instance of \threedmatch\ consists of three finite 
sets $X,Y$ and $Z$ with $T\subseteq X \times Y\times Z$, where, 
for each $(x,y,z) \in T$, $x\in X, y \in Y, z\in Z$.
And, the objective is to determine whether there exists a 
subset $M\subseteq T$ with a cardinality of $\ell$, 
where for any two distinct tuples $(x_1,y_1,z_1), (x_2,y_2,z_2) \in M$,
we have $x_1\ne x_2$, $y_1\ne y_2$, and $z_1\ne z_2$. 
Note that $|X|=|Y|=|Z|$. 
Finding a perfect 3-dimensional matching was one of the first 
problems that was proved to be \nph ({Problem~[SP1]} in \cite{DBLP:books/fm/GareyJ79}). 
We consider a restricted version of the problem where
every element appears in at most three tuples in $T$.
(See remark after Problem~[SP1] in \cite{DBLP:books/fm/GareyJ79}.)
Moreover, we rename the elements in $X \cup Y$ as follows:
Consider a vertex $z_i \in Z$ and $t^1_i, t^2_i$ and $t^3_i$
be three tuples that include $z_i$.
Note that all three tuples may not exist.
If $t^j_i = (x_{i'}, y_{i''}, z_i)$ for some $i \in [|Z|]$ and $j \in [3]$, 
then we rename vertex $x_{i'}$ as $t^{j}_i(x)$ and 
$y_{i''}$ as $t^{j}_i(y)$. 
Note that if $z_i \in Z$ appears in no tuple, then we are dealing
with a trivial \no-instance.
And, if $z_i$ appears in a unique tuple, then we need to include
that tuple in any solution, thus simplifying the instance.
Hence, it is safe to assume that every vertex appears in 
two or three tuples in $T$.

The reduction constructs an equivalent instance 
of the \sbccffa problem as follows.
See \Cref{fig:hardness} for an illustration.
\begin{itemize}
\item For every vertex $z_i\in Z$, that appears two tuples,
it adds two unique agents, say $z_{i1}, z_{i2}$, to the collection
of agents $\calA$.
Similarly, for every vertex $z_i\in Z$, that appears in three tuples,
it adds three unique agents, say $z_{i1}, z_{i2}, z_{i3}$, to the collection
of agents $\calA$. 
For notational convenience, we denote the agents in $\calA$
as $z_{ij}$ for $i \in [|Z|]$ and $j \in [3]$. 

\item For every $u\in X \uplus Y$, reduction introduces an item.
Also, for every vertex $z_i\in Z$ that appears in two tuples 
(respectively in three tuples), 
the reduction adds one dummy item $d_{i1}$ (respectively 
two dummy items $d_{i1}$ and 
$d_{i2}$). 
Let $D$ be the collection for all the items introduced in this step.
Formally, $\I=  X\uplus Y\uplus D$.
\item 
For every agent $z_{ij} \in \calA$, 
define the utility functions $\util_{ij}$ as follows:
$$
\util_{ij}(v)=
\begin{cases}
        1,~ v \in \{t_i^j(x),t_i^j(y)\}\text{ for every tuple } 
        (t_i^j(x),t_i^j(y),z_i)\in T,\\
        2,~v \in \{d_{i1},d_{i2}\} \text{ for each } z_i\in Z,\\
        0, \text{otherwise.}
\end{cases}
$$
\end{itemize}
The reduction returns $(\Co{A}, \Co{I}, \{ \util_{ij}| z_{ij} \in \calA\}, \Co{H}, s=2, \eta=2$) as the reduced instance where $\calH$ is an edgeless graph.

\begin{figure}[t]
    \centering
    \includegraphics[scale=0.8]{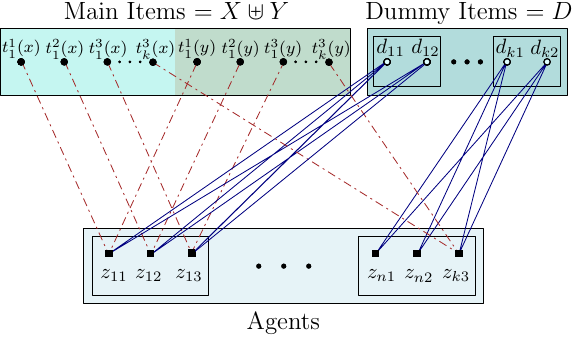}
    \caption{Illustration of Reduction for 
    Theorem~\ref{thm:np-hardness-results}$.3.b.(ii)$, i.e., when $\eta \ge 2$ and $s\ge 2$.}
    \label{fig:hardness}
\end{figure}

We now prove that both these instances are equivalent.
In the forward direction, 
if there is a $3$-dimensional perfect matching $M$ of size $|Z|$ 
with the tuples $\{t_1^1,t_2^1,\ldots, t_n^1\}\subseteq T$, 
we can construct an equivalent satisfying assignment of items to the 
agents in the \completefairdiv problem as follows:
For each tuple $t_i^1\in M$, agent $z_{i1}$ meets the utility 
requirement of two by acquiring the other two items 
(from $X \uplus Y$ in the tuple $t_i^1$), 
each endowed with a weight value of one. 
And agents $z_{i2}$ and $z_{i3}$ (if it exists) can 
acquire $d_{i1}$ and $d_{i2}$ (if it is exists), 
respectively, each with a utility value of two. 
Thus, a satisfying assignment is attained. 

In the backward direction, 
suppose a 
satisfying assignment for the \sbccffa exists, 
ensuring each agent receives a utility value of two. 
Consider agents, $z_{i1}$, $z_{i2}$ corresponding to vertex $z_i \in Z$
that appears in two tuples.
Note that the items from set $D$ can satisfy one of these agents,
the other agent requirements need to be satisfied by elements in 
$X \uplus Y$.
A similar argument holds for agent corresponding to $z_i \in Z$
that appears on three tuples.
One of the three agents' requirements corresponding to 
this vertex needs to be satisfied by items in $X \uplus Y$.
Without loss of generality, assume that $z_{i1}$s are the 
remaining agents that are not assigned any items from $D$.
But then these $n$ agents must meet their utility requirements 
using the items from $X\uplus Y$ where each agent receives 
a disjoint set from the other, consisting of two items. 
Furthermore for any agent $z_{i1}$, the two assigned items 
must be $t_i^1(x)$ and $t_i^1(y)$, the only two elements capable 
of providing positive utility value from $X\uplus Y$. 
And hence the existence of a $|Z|$ size matching 
$\{(t_i^1(x),t_i^1(y),z_i)\}_{i\in [n]}$ follows.     

We can extend the above hardness result for
any value of $\eta \ge 2$ and $s \ge 2$. 
We can introduce $|\calA|$ many isolated elements 
each with value $\eta - 2$ for every agent.
Also, note that the threshold constraints
and the number of total items ensures that 
no agent can get more than $2$ items.
Hence, the hardness holds for any value of $s \ge 2$.
Finally, once again, because of the constraints mentioned
in the previous sentence, the reduction holds
for any variation of \sbccffa.
\end{proof}



\section{Proof of Theorem~\ref{thm:fpt-nr-of-vertices}}

%
\newcommand{\polyn}[2]{\ensuremath{p^{#1}_{#2}}\xspace}
\fptnrofvertices*
\begin{proof}
The algorithm uses the technique of polynomial multiplication and fast Fourier transformation. The idea is as follows. For every agent $a_i\in \Co{A}$, where $i\in [n]$, we first construct a family of bundles that can be assigned to the agent $a_i$ in an optimal solution. Let us denote this family by $\Co{F}_i$. Then, our goal is to find $n$ disjoint bundles, one from each set $\Co{F}_i$. To find these disjoint sets efficiently, we use the technique of polynomial multiplication. 

Before we discuss our algorithm, we have to introduce some notations and terminologies.  Let $\I$ be a set of size $m$, then we can associate \I with $[m]$.  The {\em characteristic vector} of a subset $S\subseteq [m]$, denoted by $\chi(S)$, is an $m$-length vector whose $i^{\text {th}}$ bit is $1$ if and only if  $i \in S$.
Two binary strings of length $m$ are said to be disjoint if, 
for each $i\in [m]$, the $i^{th}$ bits in the two strings are different.  The {\em Hamming weight} of a binary string $S$, denoted by $H(S)$, is defined 
as the number of $1$s in the string $S$. 
A monomial $y^i$ is said to have Hamming weight $w$, 
if the degree $i$ when represented as a binary string, has Hamming weight $w$.
We begin with the following facts observed in Cygan et. al~\cite{DBLP:journals/tcs/CyganP10}:
Let $S_1$ and $S_2$ be two binary strings of the same length.
Let $S=S_1+S_2$. If $H(S)=H(S_1)+H(S_2)$, then $S_1$ and $S_2$ are disjoint binary vectors. 
Also, let $S=S_{1} \cup S_{2}$, where $S_1$ and $S_2$ are two disjoint subsets of $[m]$. 
Then, $\chi(S)=\chi(S_1)+\chi(S_2)$ and $H(\chi(S))=H(\chi(S_1))+H(\chi(S_2))=|S_1|+|S_2|$.
These two facts together yield the following. 

\begin{observation}\label{obs:HW-disjointness}
Subsets $S_{1}, S_{2} \sse \I$ are disjoint if and only if Hamming weight of the monomial $x^{\chi(S_1)+\chi(S_2)}$ is $|S_{1}|+|S_{2}|$. 
\end{observation}

The {\em Hamming projection} of a polynomial $p(y)$ to $h$, denoted by $H_{h}(p(y))$, is the sum of all the monomials of $p(y)$ which have Hamming weight $h$. We define the {\em representative polynomial} of $p(y)$, denoted by $\Co{R}(p(y))$, as the sum of all the monomials that have non-zero coefficient in $p(y)$ but have coefficient $1$ in $\Co{R}(p(y))$, i.e., it ignores the actual coefficients and only remembers whether the coefficient is non-zero. We say that a polynomial $p(y)$ {\it contains a monomial} $y^i$ if the coefficient of $y^{i}$ is non-zero. The zero polynomial is the one in which the coefficient of each monomial is $0$.

Now, we are ready to discuss our algorithm. We focus the algorithm on \pcffa and highlight the required changes for the other two variants. 

\subparagraph*{Algorithm.}
Let $\X{J}=(\Co{A}, \I, \{\util_{i}\}_{i\in [n]}, \Co{H}, \eta)$ denote an instance of \pcffa. 
We start by defining a set family indexed by the agents. 
For an agent $a_i\in \Co{A}$, let $\Co{F}_{i}$ contain each of subsets of \Co{I} that can be {\it feasibly} allocated to $a_i$ as a bundle. Specifically, a set $S \sse \Co{I}$ is in $\Co{F}_{i}$ if $S$ is an independent set in \Co{H} and the utility $\sum_{x\in S} \util_{i}(x) \geq \eta$ (subset size is at most $s$ for the size-bounded variants). We define the {\it round} \hide{of a polynomial} inductively as follows. 

For round $1$ and a positive integer $s^\star$, we define a polynomial
\[\polyn{1}{s^\star}(y) = \sum_{S\in \Co{F}_{1},\\ |S|=s^\star} y^{\chi(S)} \]

For round $i \in [n] \setminus \{1\}$, and a positive integer $s^\star$, we define a polynomial by using the $\Co{R}(\cdot)$ operator 

\[ \polyn{i}{s^\star}(y) = \sum_{\substack{S \in \Co{F}_{i} \\ s'=s^\star-|S|} } \Co{R}\left(H_{s^\star}\left(\polyn{i-1}{s'}(y) \times y^{\chi(S)}\right) \right)
\]

The algorithm returns ``yes'' if for any positive integer $s^\star \in \mathbb{Z}_{\geq 0}$ ($s^\star=m$ for complete version), the polynomial $\polyn{n}{s^\star}(y)$ is non-zero.

In fact, any non-zero monomial in the polynomial ``represents'' a solution for the instance \X{J} such that we can find the bundle to assign to each agent $a_i \in \Co{A}$ by backtracking the process all the way to round 1. 


We now describe how the algorithm computes a solution, if one exists.
We assume that for some positive integer $s^\star$, $\polyn{n}{s^\star}(y)$ 
is a non-zero polynomial. Thus, it contains a non-zero monomial of the form $\polyn{n-1}{s'}(y) \times y^{\chi(S)}$, where $S\in \Co{F}_{n}$. Note that $\chi(S)$ describes the bundle assigned to agent $a_n$, the set $S$.  Since the monomial $\polyn{n-1}{s'}(y) \times y^{\chi(S)}$ exists in the polynomial $\polyn{n}{s^\star}(y)$ after applying $H_{s^\star}(\cdot)$ function, it must be that $\polyn{n-1}{s'}(y) = y^{\chi(S')}$ for some set $S' \sse \Co{I}$ such that $S' \cap S = \emptyset$. By recursively applying the same argument to the polynomial $\polyn{n-1}{s'}(y)$, we can obtain the bundles that are allocated to  the  agents $a_{n-1},\ldots,a_1$. 

\subparagraph*{Correctness.}
We prove that the above algorithm returns ``yes'' if and only if \X{J} is a \yes-instance of \pcffa.
Suppose that \X{J} is a \yes-instance of \pcffa. Then, there is an {\it assignment}, i.e., an injective function $\phi$ that maps \items to agents. 
For each agent $a_i\in \Co{A}$, where $i\in [n]$, we define $S_{i} = \phi^{-1}(i)$. We begin with the following claim that enables us to conclude that the polynomial $\polyn{n}{s^\star}(y)$, where $s^\star=\sum_{i\in [n]}|S_i|$, contains the monomial $y^{\sum_{i\in [n]}\chi(S_i)}$. 

\begin{claim}
\label{clm:fft-fwd}
For each $j\in [n]$, the polynomial $\polyn{j}{s^\star}(y)$, where $s^\star=\sum_{i\in [j]}|S_i|$, contains the monomial $y^{\sum_{i\in [j]}\chi(S_i)}$. 
\end{claim}

\noindent \emph{Proof:}
The proof is by induction on $j$.
Consider the case when $j=1$. 
We first note that $S_1$ is in the family $\Co{F}_1$ 
as it is a feasible bundle for the agent $1$. 
Thus, due to the construction of the polynomial $\polyn{1}{s^\star}(y)$, 
we know that $\polyn{1}{|S_1|}(y)$ contains the monomial $y^{\chi(S_1)}$. 

Now, consider the induction step.
Suppose that the claim is true for $j= j'-1$. We next prove it for $j=j'$. To construct the polynomial $\polyn{j'}{s^\star}(y)$, where  $s^\star=\sum_{i\in [j']}|S_i|$, we consider the multiplication of polynomial $\polyn{j'-1}{s'}(y)$, where $s'=\sum_{i\in [j'-1]}|S_i|$, and $y^{\chi(S_{j'})}$. Due to the inductive hypothesis, $\polyn{j'-1}{s'}(y)$, where $s'=\sum_{i\in [j'-1]}|S_i|$, contains the monomial $y^{\sum_{i\in [j'-1]}\chi(S_i)}$. Note that $S_{j'}$ is in the family $\Co{F}_{j'}$ as it is a feasible bundle for the agent $j'$. Since $S_{j'}$ is disjoint from $S_1\cup \ldots \cup S_{j'-1}$, due to \Cref{obs:HW-disjointness}, we can infer that $\polyn{j'}{s^\star}(y)$, where  $s^\star=\sum_{i\in [j']}|S_i|$, has the monomial $y^{\sum_{i\in [j']}\chi(S_i)}$. \hfill $\diamond$

Due to Claim~\ref{clm:fft-fwd}, we can conclude that $\polyn{n}{s^\star}(y)$, where $s^\star=\sum_{i\in [n]}|S_i|$, contains the monomial $y^{\sum_{i\in [n]}\chi(S_i)}$. For the other direction, suppose that the algorithm returns ``yes''. Then, for some positive integer $s^\star$, $\polyn{n}{s^\star}(y)$ is a non-zero polynomial. We need to show that there exists pairwise disjoint  sets $S_1,\ldots,S_n$ such that $S_i \in \Co{F}_i$, where $i\in [n]$. This will give us an assignment function $\phi$, where for all $x\in \Co{I}$, $\phi(x)=i$, if $x\in S_i$, where $i\in [n]$.  Since each $S \in \Co{F}_i$, where $i\in [n]$, is an independent set and $\sum_{x\in S}\util_{i}(x) \geq \eta$, $\phi$ is a feasible assignment. We next prove the following claim that enables us to conclude the existence of pairwise disjoint sets. 

\begin{claim}\label{clm:fft-reverse}
For each $j\in [n]$, if the polynomial $\polyn{j}{s^\star}(y)$ is non-zero for some $s^\star \in [m]$, then there exists $j$ pairwise disjoint  sets $S_1,\ldots,S_j$ such that $S_i \in \Co{F}_i$, where $i\in [j]$.
\end{claim}

\noindent \emph{Proof:}
We prove it by induction on $j$. 
Consider the base case when $j=1$. 
Suppose $\polyn{1}{s^\star}(y)$ is non-zero for 
some $s^\star \in [m]$. 
Then, it contains a monomial $y^{\chi(S)}$, where $S\in \Co{F}_1$. Thus, the claim is true. 

Now consider the induction step. Suppose that the claim is true for $j=j'-1$. We next prove it for $j=j'$. Suppose that $\polyn{j'}{s^\star}(y)$ is non-zero for some $s^\star \in [m]$. Then, it contains a monomial of the form $\polyn{j-1}{s'}(y) \times y^{\chi(S)}$, where $|S|=s^\star -s'$ and $S\in \Co{F}_{j'}$. Due to induction hypothesis, since  $\polyn{j-1}{s'}(y)$ is a non-zero polynomial,  there exists $j'-1$ pairwise disjoint  sets $S_1,\ldots,S_{j'-1}$ such that $S_i \in \Co{F}_i$, where $i\in [j'-1]$. Furthermore, due to \Cref{obs:HW-disjointness}, we have that $S_{j'}$ is disjoint from $S_1\cup \ldots \cup S_{j'-1}$. Thus, we have $j'$ pairwise disjoint  sets $S_1,\ldots,S_{j'}$ such that $S_i \in \Co{F}_i$, where $i\in [j']$. \hfill $\diamond$

This completes the proof for \pcffa. 

To argue the correctness of the \ccffa, recall that we return yes, when $p_m^n(y)$ is a non-zero polynomial. The correctness proof is exactly the same as for \pcffa. Note that it also remains the same for the size-bounded variant, as we have taken care of size while constructing families for each agent.

\subparagraph*{Running time.}
We argue that the algorithm runs in $\Oh(2^{m}\cdot(n+m)^{\Oh(1)})$ time. 
In the algorithm, we first construct a family of feasible bundles 
for each agent $a_i\in \Co{A}$. Since we check all the subsets of $\I$, 
the constructions of families takes $\Oh(2^m\cdot (n+m)^{\Oh(1)})$ time. 
For $i=1$, we construct $m$ polynomials that contains $\Oh(2^m)$ terms. 
Thus, $p_s^1(y)$ can be constructed in $\Oh(2^m\cdot m)$ time. 
Then, we recursively construct polynomials by polynomial multiplication. 
Recall that there exists an algorithm that multiplies two polynomials of degree $d$ in $\Oh(d \log d)$ time~\cite{moenck1976practical}.
Since every polynomial has the degree at most $\Oh(2^m)$, 
every polynomial multiplication takes $\Oh(2^m \cdot m)$ time. 
Hence, the algorithm runs in  $\Oh(2^{m}\cdot(n+m)^{\Oh(1)})$ time. 


We first consider the complete variant of the problem. To prove the lower bound for \ccffa, we give a polynomial time reduction from the \textsc{$3$-Coloring} problem, in which given a graph $G$, the goal is to check whether the vertices of $G$ can be colored using $3$ colors so that no two neighbors share the same color.   
For a given instance $G$ of \textsc{$3$-Coloring},    we construct an instance of \ccffa as follows: $\calA=\{a_1,a_2,a_3\}$, $\calI=V(G)$ and $G$ is the conflict graph. For each $i\in [3]$ and $x\in \calI$, $\util_i(x)=1$, and
the target value $\eta = 1$. 
It is easy to see that this simple reduction constructs 
an equivalent instance. 
Since, \textsc{$3$-Coloring} does not admit 
an algorithm running in time $2^{o(m)}$ unless the \ETH\ fails, we get the lower bound result
(See, for example, Chapter~14 in \cite{cygan2015parameterized}).

Next, we move to the partial variant of the problem. Towards this, we give a polynomial time reduction from the \textsc{Independent Set} problem. 
Given an instance $\X{I}=(G, k)$ of {\sc Independent Set}, we create an equivalent instance ${\cal J}=(\calA, \calI, \{\util_{i}\}_{i\in [n]}, G, \eta)$ 
of the \partialfairdiv problem as follows:
 $\calA=\{a_1\}, \calI=V(G)$, $\util_1(x)=1$ for each \items $x\in \calI$. The target value $\eta$ is set to $k$.
It is important to note that any independent set in graph $G$ 
of size at least $k$ is simultaneously a 
valid allocation for the sole agent in the reduced instance. 
The converse is also true, meaning that any appropriate assignment 
to the single agent corresponds to an independent set of size at least $k$ in $G$. 
Since \textsc{Independent Set} does not admit
an algorithm that runs in time $2^{o(m)}$ unless the \ETH\ fails, our claim holds true
(Once again, see Chapter~14 in \cite{cygan2015parameterized}.)
Observe that the same reduction works for \sbpcffa as well.

Moving forward, we focus on \sbccffa and we give a similar reduction from {\sc Independent Set}. Given an instance $\X{I}=(G, k)$ of {\sc Independent Set}, we create an equivalent instance ${\cal J}=(\calA, \calI, \{\util_{i}\}_{i\in [n]}, G, s, \eta)$ 
of the \sbccffa problem as follows:
 $\calA=\{a_1,a_2, \cdots, a_{m-k+1}\}, 
 ~\calI=V(G)$, $\util_1(x)=1$ for each \items $x\in \calI$,  and $\util_{i\in [n]\setminus \{1\}}(x)=k, ~i\in [n]\setminus \{1\}$ for each \items $x\in \calI$. The target value $\eta$ is set to $k$ and $s=k$. Note that any independent set of size $k$ provides a satisfying assignment to the agent $a_1$. The remaining $(m-k)$ \items can satisfy the other $(m-k)$ agents in the reduced instance. The reverse direction follows from the fact that any valid satisfying assignment to the agent $a_1$ corresponds to an independent set of size $k$ in $G$. Thus, the theorem is proved. 
\end{proof}

\section{Proof of Theorem~\ref{thm:equiavlence} and Theorem~\ref{lem:mwis-degenerate}}

\begin{sloppypar}
We begin with the proof of  Theorem~\ref{thm:equiavlence}. 
\equivalence*
Let $\X{J}=(\Co{A}, \I, \{\util_{i}\}_{i\in [n]}, \Co{H}, s,\eta)$  be an instance of \cffag, and let $|{\mathscr J}|$ denote the size of the instance.  We first prove the first part of Theorem~\ref{thm:equiavlence}, which is the easier direction of the proof.  In particular, let $\mathbb{A}$ be an \fpt algorithm for  \cffag, running in time $f(n,s) \ptime$. Given an instance $(G,k,\rho,w)$ of \mwisg, we construct an instance $\X{J}=(\Co{A}, \I, \{\util_{i}\}_{i\in [n]}, \Co{H},s, \eta)$   of \cffag as follows. The set of agents $\Co{A}$ has only one agent $a_1$. Further, $\I=V(G)$, $\util_{1}=w$,  $\Co{H}=G$, $s=k$, and $\eta=\rho$. It is easy to see that 
$(G,k,\rho,w)$ is a \yes-instance of \mwisg if and only if $\X{J}$  is a  \yes-instance of \cffag. Thus, by invoking algorithm $\mathbb{A}$ on instance $\X{J}$ of \cffag, we get an \fpt algorithm for \mwisg that runs in $f(k) \ptime$ time. This completes the proof in the forward direction. Next, we prove the reverse direction of the proof. That is, given an \fpt algorithm for \mwisg, we design an \fpt algorithm for \cffag.  For ease of explanation, we first present a randomized algorithm which will be derandomized later using the known tool of {\em $(p,q)$-perfect hash family}~\cite{alon1995color,fomin2014efficient}. 
\end{sloppypar}

%

\subparagraph{Randomized Algorithm.}\label{sec:randomized algo}
We design a randomized algorithm with the following specification. If the input, \X{J}, is a  \no-instance then the algorithm always returns ``no''. However,   if the input, \X{J}, is a  \yes-instance then the algorithm returns ``yes'' with probability at least $1/2$.

We throught assume that we have been given a \yes-instance. This implies that 
  there exists a  hypothetical solution $\phi \colon \Co{I}  \rightarrow \Co{A}$. We define everything with respect to $\phi$. That is, $\phi \colon \Co{I}  \rightarrow \Co{A}$ is a partial assignment satisfying all the requirements. Let $S=\phi^{-1}(\Co{A})=\cup_{a\in \Co{A}}\phi^{-1}(a)$, i.e., the set of \objects that are assigned to some agent.  Further, note that $|S|\leq ns$, as the size of each bundle is upper bounded by $s$. 
Our main idea is to first highlight all the \objects in the set $S$, that are assigned to some agent, using color coding as follows: color the vertices of  $\Co{H}$ uniformly and independently at random using $ns$ colors, say $\{1,\ldots,ns\}$.   

%
The goal of the coloring is that ``with high probability'', we color the \objects assigned to agents in a solution using distinct colors. The following proposition bounds the success probability. 
%

\begin{proposition}{\rm \cite[Lemma 5.4]{ParamAlgorithms15b}}\label{prop:success-prob} Let $U$ be a universe and $X\subseteq U$. Let $\chi \colon U \rightarrow [|X|]$ be a function that colors  each element of $U$ with one of $|X|$ colors uniformly and independently at random. Then, the probability that the elements of $X$ are colored with pairwise distinct colors is at least $e^{-|X|}$.
\end{proposition}

Due to \Cref{prop:success-prob}, the coloring step of the algorithm colors the \objects in $\phi^{-1}(\Co{A})$ using distinct colors with probability at least $e^{-ns}$. We call an assignment $\phi \colon \Co{I}  \rightarrow \Co{A}$ as {\em colorful} if every two \items $\{i, i'\} \in \phi^{-1}(\Co{A})$ get distinct color. Moreover, for each agent $a$, $|\phi^{-1}(a)|\leq s$.


Next, we find a {\em colorful} feasible assignment in the following lemma.  Further, let us assume that we have an \fpt algorithm, $\mathbb{B}$, for \mwisg  running in time $h(k)n^{\Oh(1)}$.

\begin{lemma}\label{lem:colorful_solution}
Let  $\X{J}=(\Co{A}, \I, \{\util_{i}\}_{i\in [n]}, \Co{H}, s,\eta)$
be an instance of \cffag and $\chi \colon V(\Co{H}) \rightarrow [ns]$ be a coloring function. Then, there exists a dynamic programming algorithm that finds a colorful feasible assignment $\phi \colon \Co{I}  \rightarrow \Co{A}$ in $\Oh(3^{ns}\cdot h(s) \cdot (n+m)^{\Oh(1)})$ time, if it exists, otherwise, return ``no''.
\end{lemma}

\begin{proof}
Let $\mathsf{colors} = \{1,\ldots,ns\}$ be the set of colors and let $a_1,\ldots,a_n$ be an arbitrary order of the agents. For $S\subseteq \mathsf{colors}$, let $V_S$ be the subset of vertices in $\Co{H}$ that are colored using the colors in $S$. We apply dynamic programming: for a non-empty set $S\subseteq \mathsf{colors}$ and $i\in [n]$, we define the table entry $T[i,S]$ as $1$ if there is a feasible assignment of \objects 
in $V_S$ 
to agents $\{a_1,\ldots,a_i\}$; otherwise it is $0$.  Note that we are not {\em demanding} that the feasible assignment be colorful. For an agent $a_i\in \Co{A}$ and $S\subseteq \mathsf{colors}$, let $\Co{H}_{a_i,S}$ be a vertex-weighted graph constructed as follows. Let $V_S$ be the subset of vertices in $\Co{H}$ that are colored using the colors in $S$. Then,  $\Co{H}_{a_i,S}=\Co{H}[V_S]$. 
The weight of every vertex $x \in \Co{H}_{a_i,S}$ is $\util_{i}(x)$. For a vertex-weighted graph $G$, let $\mathbb{I}(G)\in \{0,1\}$, where 
$\mathbb{I}(G)=1$ if there exists an independent set of size at most $s$ and weight at  least $\eta$ in $G$, otherwise $0$. We compute $\mathbb{I}(G)$ using algorithm $\mathbb{B}$.
%
 We compute the table entries as follows. 
 
\noindent{Base Case: } For $i=1$ and non-empty set $S$, we compute as follows:
\begin{equation}\label{eq:base case}
T[1,S] = \mathbb{I}(\Co{H}_{a_1,S})
\end{equation}

\noindent{Recursive Step: } For $i>1$ and non-empty set $S$, we compute as follows:
\begin{equation}\label{eq:recursion}
T[i,S] = \bigvee_{\emptyset \neq S' \subset S} T[i-1,S'] \wedge \mathbb{I}(\Co{H}_{a_i,S\setminus S'})
\end{equation}

We return ``yes'' if $T[n,S]=1$ for some $S\subseteq \mathsf{colors}$, otherwise ``no''. 
Next, we prove the correctness of the algorithm. Towards this, we prove that the \Cref{eq:base case} and \Cref{eq:recursion} correctly compute $T[i,S]$, for each $i\in [n]$ and $\emptyset \neq S \subseteq \mathsf{colors}$.

\begin{claim}\label{clm:correctness-equations}
\Cref{eq:base case} and \Cref{eq:recursion} correctly compute $T[i,S]$, for each $i\in [n]$ and $\emptyset \neq S \subseteq \mathsf{colors}$. 
\end{claim}

\noindent \emph{Proof:}
We will prove it by induction on $i$. For $i=1$, we are looking for any feasible assignment of \objects colored using the colors in $S$ to the agent $a_1$. Thus, \Cref{eq:base case} computes $T[1,S]$ correctly due to the construction of the graph $\Co{H}_{a_1,S}$ and the correctness of algorithm $\mathbb{B}$.  

Now, consider the recursive step. For $i>1$ and $\emptyset \neq S \subseteq \mathsf{colors}$, we compute $T[i,S]$ using \Cref{eq:recursion}. We show that the recursive formula is correct. Suppose that  \Cref{eq:recursion} computes $T[i',S]$ correctly, for all $i'<i$ and   $\emptyset \neq S \subseteq \mathsf{colors}$. First, we show that $T[i,S]$ is at most the R.H.S. of  \Cref{eq:recursion}. If $T[i,S]=0$, then the claim trivially holds. Suppose that $T[i,S]=1$. Let $\psi$ be a colorful feasible assignment to agents $\{a_1,\ldots,a_i\}$ using \objects that are colored using colors in $S$. Let $S_j\subseteq S$ be the set of colors of \objects in $\psi(a_j)$, where $j\in [i]$.  Since $\psi(a_i)$ uses the colors from the set $S_i$ and $\sum_{x \in \psi(a_i)}\util_{a_i}(x) \geq \eta$, due to the construction of $\Co{H}_{a_i,S_i}$, we have that  $\mathbb{I}(\Co{H}_{a_i,S_i})=1$.   Consider the assignment $\psi'=\psi\vert_{\{a_1,\ldots,a_{i-1}\}}$ (restrict the domain to $\{a_1,\ldots,a_{i-1}\}$). Since  $S_i$ is disjoint from $S_1\cup \ldots \cup S_{i-1}$ due to the definition of colorful assignment, $\psi'$ is a feasible assignment for the agents $\{a_1,\ldots,a_{i-1}\}$ such that the color of all the \objects in $\psi'(\{a_1,\ldots,a_{i-1}\})$ is in $S\setminus S_i$. Furthermore, since $\psi$ is colorful, $\psi'$ is also colorful. Therefore, $T[i-1,S\setminus S_i]=1$ due to the induction hypothesis. Hence, R.H.S. of \Cref{eq:recursion} is $1$. Thus, $T[i,S]$ is at most R.H.S. of  \Cref{eq:recursion}.

For the other direction, we show that $T[i,S]$ is at least R.H.S. of  \Cref{eq:recursion}. If R.H.S. is $0$,  then the claim trivially holds. Suppose R.H.S. is $1$. That is, there exist $S' \subseteq S$ such that $T[i-1,S'] = 1$ and $\mathbb{I}(\Co{H}_{a_i,S\setminus S'})=1$.  Let $\psi$ be a feasible assignment to agents $\{a_1,\ldots,a_{i-1}\}$ using \objects that are colored using the colors in $S'$. Since $\mathbb{I}(\Co{H}_{a_i,S\setminus S'})=1$, there exists a subset $X\subseteq V_{S\setminus S'}$ such that $\sum_{x \in X}\util_{a_i}(x) \geq \eta$.  Thus, construct an assignment $\psi'$ as follows: $\psi'(a)=\psi(a)$, if $a\in \{a_1,\ldots,a_{i-1}\}$ and $\psi'(a_i)=X$. Since $\psi'$ is a feasible assignment and $\mathbb{I}(\Co{H}_{a_i,S\setminus S'})=1$, $\psi$ is a feasible assignment. 
Therefore, $T[i,S]=1$. \hfill $\diamond$
%
Therefore $T[n,S]=1$  for some $S\subseteq \mathsf{colors}$  if and only if $\X{J}$ is a yes-instance of \cffag. This completes the proof of the lemma.
 \end{proof}


Due to \Cref{prop:success-prob} and \Cref{lem:colorful_solution}, we obtain an $\Oh(3^{ns}\cdot h(s) \cdot (n+m)^{\Oh(1)})$ time randomized algorithm for \cffag which succeeds with probability $e^{-ns}$. Thus, by repeating the algorithm independently $e^{ns}$ times, we obtain the following result.

\begin{sloppypar}
\begin{lemma}\label{thm:randomized_algo}
There exists a randomized algorithm that given an instance $\X{J}=(\Co{A}, \I, \{\util_{i}\}_{i\in [n]]}, \Co{H}, s,\eta)$ of \cffag either reports a failure or finds a feasible assignment in $\Oh((3e)^{ns}\cdot h(s) \cdot (n+m)^{\Oh(1)})$ time. Moreover, if the algorithm is given a yes-instance, the algorithm returns ``yes'' with  probability at least $1/2$, and if the algorithm is given a no-instance, the algorithm returns ``no'' with  probability $1$. 
\end{lemma}
\end{sloppypar}

\begin{proof}
Let $\X{J}=(\Co{A}, \I, \{\util_{i}\}_{i\in [n]}, \Co{H}, s,\eta)$ be an instance of \cffag. We color the \objects uniformly at random with colors $[ns]$. Let $\chi\colon V(\Co{H}) \rightarrow [ns]$ be this coloring function. We run the algorithm in \Cref{lem:colorful_solution} on the instance $\X{J}$ with coloring function $\chi$. If the algorithm returns ``yes'', then we return ``yes''. Otherwise, we report failure.  

Let $\X{J}$ be a yes-instance of \cffag and $\phi$ be a hypothetical solution. Due to \Cref{prop:success-prob}, all the \objects in $\phi^{-1}(\Co{A})$ are colored using distinct colors with probability at least $e^{-ns}$. Thus, the algorithm in \Cref{lem:colorful_solution} returns yes with probability at least $e^{-ns}$. Thus, to boost the success probability to a constant, we repeat the algorithm independently $e^{ns}$ times. Thus, the success probability is at least 
\begin{equation*}
1-\Big(1-\frac{1}{e^{ns}}\Big)^{ns} \geq 1-\frac{1}{e} \geq \frac{1}{2}
\end{equation*}

If the algorithm returns ``yes'', then clearly $\X{J}$ is a yes-instance of \cffag due to \Cref{lem:colorful_solution}. 
 \end{proof}

\subparagraph{Deterministic Algorithm.}\label{sec:deterministic algo}
We derandomize the algorithm 
using $(p,q)$-perfect hash family to obtain a deterministic algorithm for our problem. 

\begin{defn}[$(p,q)$-perfect hash family]{\rm (\cite{alon1995color})}
For non-negative integers $p$ and $q$, a family of functions $f_1,\ldots,f_t$ from a universe $U$ of size $p$ to a universe of size $q$ is called a $(p,q)$-perfect hash family, if for any subset $S\subseteq U$ of size at most $q$, there exists $i\in [t]$ such that $f_i$ is injective on $S$. 
\end{defn}

We can construct a $(p,q)$-perfect hash family using the following result. 

\begin{proposition}[\cite{naor1995splitters,ParamAlgorithms15b}]\label{prop:hash family construction}
There is an algorithm that given $p,q \geq 1$ constructs a $(p,q)$-perfect hash family of size $e^qq^{\Oh(\log q)}\log p$ in time $e^qq^{\Oh(\log q)}p \log p$.
\end{proposition}

Let $\X{J}=(\Co{A}, \I, \{\util_{i}\}_{i\in [n]}, \Co{H}, s,\eta)$ be an instance of \cffag. Instead of taking a random coloring $\chi$, we construct an $(m,ns)$-perfect hash family $\Co{F}$ using \Cref{prop:hash family construction}. Then, for each function $f\in \Co{F}$, we invoke the algorithm in \Cref{lem:colorful_solution} with the coloring function $\chi=f$. If there exists a feasible assignment $\phi \colon \Co{I} \rightarrow \Co{A}$ such that $|\phi(a)|\leq s$, for all $a\in \Co{A}$, then there exists a function $f\in \Co{F}$ that is injective on $\phi^{-1}(\Co{A})$, since $\Co{F}$ is an $(m,ns)$-perfect hash family. Consequently, due to \Cref{lem:colorful_solution}, the algorithm return ``yes''. Hence, we obtain the following deterministic algorithm.

\begin{theorem}\label{thm:fpt-ns}
There exists a deterministic algorithm for  \cffag running in time 
$\Oh((3e)^{ns}\cdot (ns)^{\log ns} \cdot h(s)\cdot (n+m)^{\Oh(1)})$.
\end{theorem}

Due to \Cref{thm:fpt-ns}, we can conclude the following. 

\begin{corollary}\label{cor1:fpt-ns}
If \mwisg is solvable in polynomial time, then there exists a deterministic algorithm for  \cffag running in time 
$\Oh((3e)^{ns}\cdot (ns)^{\log ns} \cdot (n+m)^{\Oh(1)})$. 
%
%
\end{corollary}

It is possible that {\sc MWIS} is polynomial-time solvable on $\mathcal G$, but  \mwisg is \npc, as {\em any} $k$-sized solution of \mwis need not satisfy the weight constraint in the \mwisg problem. However, when we use an algorithm 
$\mathbb{B}$ for  \mwisg  in our algorithm, we could have simply used an algorithm for {\sc MWIS}. Though this will not result in a size bound of $s$ on the size of an independent set of weight at least $\eta$ that we found, however, this is sufficient to solve \pcffa, not \sbpcffa though. However, we need to use \mwisg, when  {\sc MWIS} is \npc and we wish to use \fpt algorithm with respect to $k$. 
%
%
Due to \Cref{thm:fpt-ns} and this observation, \pcffa is \fpt when parameterized by $n+s$ for several graph classes, such as chordal graphs~\cite{golumbic2004algorithmic},  bipartite graphs~\cite{golumbic2004algorithmic}, $P_6$-free graphs~\cite{grzesik2022polynomial}, outerstring graph~\cite{keil2017algorithm}, and fork-free graph~\cite{lozin2008polynomial}.

\begin{remark}
Our algorithm for chordal graphs is an improvement over the known algorithm that runs in $\Oh(m^{n+2}(Q+1)^{2n})$ time, where $Q=\max_{a\in \Co{A}}\sum_{i\in \Co{I}}p_{a}(i)$~{\rm \cite{DBLP:conf/iwoca/ChiarelliKMPPS20}}.
\end{remark}

Next, we prove \Cref{lem:mwis-degenerate}. 
%
\fptfrindlygraph*
\begin{proof}
Let $(G,k,\rho,w)$ be a given instance of \mwisg. Further, $\mathcal G$ is  $f$-ifc. 
Let ${\sf HighWeight} = \{v\in V(G) \colon w(v) \geq \nicefrac{\rho}{k}\}$. Note that if there exists an independent set of $G[{\sf HighWeight}]$ of size $k$, then it is the solution of our problem. Since, $G$ belongs to $f$-ifc, $G[{\sf HighWeight}]$ is also $f$-ifc. Thus, there exists an independent set in $G[{\sf HighWeight}]$ of size at least  $f(|{\sf HighWeight}|)$.  
If $f(|{\sf HighWeight}|) \geq k$, then there exists a desired solution in $G[{\sf HighWeight}]$. To find a solution, we do as follows. Consider an arbitrary  set 
$X \subseteq {\sf HighWeight}$ of size $f^{-1}(k)$. The size of $X$  guarantees that the set $X$ also has a desired solution. Now  we enumerate  subsets of size $k$ of $X$ one by one, and check whether it is independent; and if independent return it. This concludes the proof. 
 Otherwise, $|{\sf HighWeight}|< f^{-1}(k)$.
Note that the solution contains at least one vertex of ${\sf HighWeight}$. Thus, we guess a vertex, say $v$, in the set ${\sf HighWeight}$ which is in the solution,  delete $v$ and its neighbors from $G$, and decrease $k$ by $1$. Repeat the algorithm on the instance $(G-N[v],k-1,\rho-w(v),w\lvert_{V(G-N[v])})$. 

Since the number of guesses at any step of the algorithm is at most $f^{-1}(k)$ and the algorithm repeats at most $k$ times, the running time of the algorithm is $\Oh((f^{-1}(k))^k\cdot(n+m)^{\Oh(1)})$.
\end{proof}
Hence we have the following corollary.
\begin{corollary}\label{cor:fpt-k-mwisg}
There exists an algorithm that solves \mwisg in $\Oh((2k)^k\cdot(n+m)^{\Oh(1)})$, $\Oh((4 k^2)^k\cdot(n+m)^{\Oh(1)})$, $\Oh((4k)^k\cdot(n+m)^{\Oh(1)})$, $\Oh((dk+k)^k\cdot(n+m)^{\Oh(1)})$, $\Oh(R(\ell,k)^k\cdot(n+m)^{\Oh(1)})$  time, when $\Co{G}$ is a family of bipartite graphs, triangle free graphs, planar graphs, $d$-degenerate graphs, graphs excluding $K_\ell$ as an induced graphs, respectively. Here, $R(\ell,k)$ is an upper bound on Ramsey number. 
\end{corollary}

\section{Proof of Theorem~\ref{thm:fpt-nr-of-non-edges}}

\fptnonedges*
\begin{proof}~
\begin{enumerate}
\item Subexponential time algorithm  when the number of agents is constant.
\end{enumerate}
The algorithm is based on counting the independent sets of size at least two and then mapping them to the agents. Thus, the algorithm works for all three variants of \cffa. Hence, for convenience, we write \cffa in this proof.

We first observe that the complement graph of \Co{H}, denoted by $\overline{\Co{H}}$, contains all the vertices of \Co{H} but $t$ edges only. Moreover, each clique in this graph constitutes a \conflict-free bundle in the instance of \cffa. Conversely, we claim that any \conflict-free bundle in the instance of \cffa must form a clique in $\overline{\Co{H}}$ since for every pair of \items $x_{1}, x_{2}$ in a bundle, there exists an edge in $\overline{\Co{H}}$. 

Thus, enumerating all possible cliques (not just maximal ones) in $\overline{\Co{H}}$ allows us to check for possible allocations to agents. To show that this is doable in the claimed time, we will count the number of cliques in $\overline{\Co{H}}$. Since $\overline{\Co{H}}$ has $t$ edges, there can be at most $2t$ vertices that are not isolated. Vertices that are isolated {\em constitute a clique of size $1$}, and are called {\it trivial cliques}. They are upper bounded by the number of \items ($m$), and will be counted separately. A clique is said to be {\em non-trivial} if it does not contain an isolated vertex. Next, we will upper bound the non-trivial cliques. Towards this, we first show that $\overline{\Co{H}}$ is a $2\sqrt{t}$-degenerate graph by a simple counting argument. Note that if there exists a subgraph $H$ with minimum degree at least $2\sqrt{t}$, then the graph must have more than $t$ edges. Let $H$ be the subgraph of $\Co{H}$ induced on the non-isolated vertices of $\Co{H}$. Since $H$ has at most $t$ edges, every subgraph of $H$ has a vertex of degree at most $2\sqrt{t}$. Thus, $H$ is a $2\sqrt{t}$-degenerate graph, and hence has a $2\sqrt{t}$-degeneracy sequence. 
%
%
%
%
%
%
%
%
%


Let $\Co{D}= v_{1}, \ldots, v_{2t}$ denote a $2\sqrt{t}$-degenerate degree sequence of $H$. 
Notice that for any $i\in [2t]$, $v_{i}$ has at most $2\sqrt{t}$ neighbors among $\{v_{j}\colon  j>i\}$. Consider the  $2\sqrt{t}$ neighbors of $v_{1}$ and among them there can be at most $2^{2\sqrt{t}}$ cliques and can be enumerated in time $\Oh(2^{2\sqrt{t}})$. By iterating over $v_{i}$, we can enumerate all the non-trivial cliques  in $\overline{\Co{H}}$ in $\Oh(2t \cdot 2^{2\sqrt{t}})$ time. Indeed, for a non-trivial clique $C$, if $v_i$ is the first vertex in $C$ with respect to $\Co{D}$, that is all other vertices in $C$ appear after $v_i$ in  $\Co{D}$, then $C$ is enumerated when we enumerate all the cliques with respect to $v_i$ in our process.  This implies that the number of independent sets in   \Co{H}  
is upper bounded by $\Oh(2t \cdot 2^{2\sqrt{t}} + m) $ and the number of independent sets of size at least $2$ in   \Co{H}   is upper bounded by  $\Oh(2t \cdot 2^{2\sqrt{t}} ) $. Let $\mathbb{I}_{\geq 2}$ denote the family of independent sets of \Co{H} that have size at least $2$ -- the family of non-trivial independent sets. 

Thus, one potential algorithm is as follows. We first guess which agents are assigned non-trivial independent sets and which independent set.  That is, for each agent $a\in \Co{A}$, we guess  an independent set $I_a\in \mathbb{I}_{\geq 2} \cup \gamma$ ($\gamma$ is just to capture that the agent will not get non-trivial bundle). Let $\Co{A}' \subseteq \Co{A}$ be the set of agents for whom the guess is not $\gamma$. Let $(\Co{A}', \{I_a\}_{a\in\Co{A}'}) $ denote the corresponding guess for the agents in $\Co{A}'$. 
We first check that the guess for $\Co{A}' $  is {\em correct}. Towards that we check that for each $a_1,a_2 \in \Co{A}'$, $I_{a_1}\cap I_{a_2}=\emptyset$ and for each $a \in \Co{A}'$, we have that $\sum_{i\in I_a}\util_{a}(i) \geq \eta$. For the complete variant, we additionally need to check that $\cup_{a\in {\Co{A}}}I_a$ is the set of all non-isolated vertices in $\overline{\Co{H}}$. Since, $|\mathbb{I}_{\geq 2}|$ is upper bounded by $\Oh(2t \cdot 2^{2\sqrt{t}})$, the number of guess are upper bounded by  $\Oh((2t \cdot 2^{2\sqrt{t}}+1)^n)$. For each correct  guess $(\Co{A}', \{I_a\}_{a\in\Co{A}'}) $, we solve the remaining problem in polynomial time by invoking the algorithm for \sbpcffa or \sbccffa for $s=1$ in Theorem~\ref{thm:np-hardness-results}, depending on whether we are solving complete variant or partial. 
Let $\Co{A}^\star=\Co{A}\setminus \Co{A}'$, $\Co{I}^\star=\Co{I}\setminus (\bigcup_{a\in\Co{A}'} I_a)$, and $s=1$.   Then, we apply Theorem~\ref{thm:np-hardness-results} on the following instance: $(\Co{A}^\star,\Co{I}^\star,(\util_i)_{a_i \in \Co{A}^\star},\Co{H}[\Co{I}^\star],s,\eta)$, here $\Co{H}[\Co{I}^\star]$ is a clique. This implies that the total running time of the algorithm is upper bounded by  $\Oh((2t \cdot 2^{2\sqrt{t}}+1)^n(n+m)^{\Oh(1)})$. 

%

\begin{enumerate}[resume]
\item An \fpt algorithm with respect to $t$.
\end{enumerate}
Since this algorithm also works for all three variants of \cffa, we write \cffa in this proof, for convenience.

Let $\X{J}=(\Co{A}, \I, \{\util_{i}\}_{i\in [n]}, \Co{H}, \eta)$ be a given instance of \cffa.  
Let $V_{>1}$ be the set of vertices which are part of independent sets of size at least $2$. As argued for the subexponential algorithm, $|V_{>1}|\leq 2t$. Thus, there are at most $t$ bundles that contains more than one \object.  We guess partition of the \objects in $V_{>1}$  into at most $t+1$ sets, $\mathsf{notLarge, Large_1, \ldots, Large_\ell}$, where $\ell \leq t$, such that each set is an independent set in $\Co{H}$. The set $\mathsf{notLarge}$ might be empty. This contains the set of \objects in $V_{>1}$ which will not be part of any bundle of size at least $2$. The size of $\mathsf{Large}_j$ is at least $2$, for every $j\in [\ell]$, and each $\mathsf{Large}_j$ will be assigned to distinct agents in the solution (the size should be adjusted for size-bounded variant). Next, we construct a complete graph $\Co{H}'$ as follows. For each $\mathsf{Large}_j$, where $j\in [\ell]$, we have a vertex $\mathsf{Large}_j$ in $\Co{H}'$, and $\util'_i(\mathsf{Large}_j)=\sum_{x\in {\mathsf{Large}_j}}\util_i(x)$, where $i\in [n]$. If a vertex $v \in \Co{H}$ does not belong to any $\mathsf{Large}_j$, where $j\in [\ell]$, then add the vertex $v$ to $\Co{H}'$, and $\util'_i(v)=\util_i(v)$, where $i\in [n]$. Let $\X{J}'=(\Co{A}, \I, \{\util'_{i}\}_{i\in [n]}, \Co{H}', \eta)$ be the new instance of \cffa where $\Co{H}'$ is a complete graph. We again invoke the polynomial time algorithm for \sbpcffa or \sbccffa for $s=1$ in Theorem~\ref{thm:np-hardness-results}, depending on whether we are solving complete variant or partial, on the instance $\X{J}'$ with $s=1$, and return the ``yes'', if the solution exists.   
If the algorithm does not find the assignment for any guessed partition, then we return ``no''. The running time follows from \Cref{thm:np-hardness-results} and the fact that there are at most $(2t)^{t+1}$ possible partitions. 

Next, we prove the correctness of the algorithm. Suppose that $\X{J}$ is a yes-instance of \cffa and $\phi$ be one of its solution. Let $\Co{B}=\{\phi^{-1}(a)\colon a\in \Co{A} \text{ and } |\phi^{-1}(a)|\geq 2\}$. Clearly, sets in $\Co{B}$ are disjoint subsets of $V_{>1}$. Let $\Co{B}=\{B_1,\ldots,B_\ell\}$. Let $X\subseteq V_{>1}$ contains all the \objects that do not belong to any set in $\Co{B}$.  Since we try all possible partitions of $V_{>1}$, we also tried $B_1,\ldots,B_\ell,X$. Without loss of generality, let $B_i$ is assigned to $a_i$ under $\phi$. Thus, in the graph $\Co{H}'$, there is a matching $M=\{a_iB_i \in i\in [\ell]\} \cup \{a\phi^{-1}(a) \colon |\phi^{-1}(a)|=1\}$ that saturates $\Co{A}$ in the proof of Theorem~\ref{thm:np-hardness-results} for $s=1$. Thus, the algorithm returns ``yes''. The correctness of the other direction follows from the correctness of Theorem~\ref{thm:np-hardness-results} for $s=1$ case and the construction of the instance $\Co{J}'$.
\end{proof}




\section{Proof of \Cref{thm:twocliques}}
\twocliques*
\begin{proof}~
\begin{enumerate}
\item Polynomial time algorithm.
\end{enumerate}
\subparagraph{For the partial variant.}
We first present the algorithm for partial variant. Note that here $s\leq 2$. Theorem~\ref{thm:np-hardness-results} consider the case $s=1$. Thus, here we can assume that $s=2$. Hence, \pcffa and \sbpcffa variants are same in this case.

Let $\X{J}=(\Co{A}, \I, \{\util_{i}\}_{i\in [n]}, \Co{H}, \eta)$  be a given instance of \pcffa. Since the utility functions are uniform, we skip the agent identification from the subscript of utility function, i.e., instead of writing $\util_{i}$ for the utility function of agent $a_i$, we will only use $\util$. 

We note that if there exists a \object $z$ such that $\util(z) \geq \eta$, then there exists a  solution that assign it to some agent. Since the utility functions are uniform, it can be assigned to any agent. Let $\Co{I}_{\sf HighUtility}\subseteq \Co{I}$ be the set of \items whose utility is at least $\eta$, i.e., $\Co{I}_{\sf HighUtility} =\{z\in \Co{I} \colon \util(z)\geq \eta\}$. Let $\Co{I}_{\sf LowUtility}= \Co{I}\setminus \Co{I}_{\sf HighUtility}$. If  $|\Co{I}_{\sf HighUtility}|  \geq n$, then every agent get an \object from the set  $\Co{I}_{\sf HighUtility}$, and it is a solution. Otherwise, there are $|\Co{A}|-|\Co{I}_{\sf HighUtility}|$ agents to whom we need to assign bundles of size two. Let ${\sf IS}$ denote the set of all independent sets of size two in $\Co{H}[\Co{I}_{\sf LowUtility}]$. Thus, ${\sf IS}$ has size at most $m^2$. 

Next, we construct a graph, denoted by  $\Co{\widehat{H}}$,  on the \items in $\Co{I}_{\sf LowUtility}$ where there is an edge between vertices $a$ and $b$ if $\{a, b\} \in {\sf IS}$ and 
$u(a) + u(b)\geq \eta$. In this graph we compute a maximum sized matching, denoted by \Co{M}. If its size is less than  $n-|\Co{I}_{\sf HighUtility}|$, then we return the answer ``no''. Otherwise, we return answer ``yes'' and create an assignment as follows: if $(a,b) \in \Co{M}$, then we have a bundle containing $\{a,b\}$. We create exactly $n - |\Co{I}_{\sf HighUtility}|$ such bundles of size two and discard the others. These bundles along with the singleton bundles from $\Co{I}_{\sf HighUtility}$ yield our assignment for $n$ agents.

Clearly, this graph has $m$ vertices and at most $m^2$ edges. Thus, the maximum matching can be found in polynomial time. Next, we prove the correctness of the algorithm.

\noindent{\bf Correctness:} If the algorithm returns an assignment of \items to the agents, then clearly, for every agent the utility from the bundle is at least $\eta$. Every bundle is also an independent set in $\Co{H}$. Moreover, if a bundle is of size one, then the singleton \object is clearly an element of the set  $\Co{I}_{\sf HighUtility}$; otherwise, the bundle represents an independent set of size two in {\sf IS} whose total utility is 
 at least $\eta$. There are $n$ bundles in total,  exactly $|\Co{I}_{\sf HighUtility}| $ bundles of size one and at least $n - |\Co{I}_{\sf HighUtility}|$ bundles of size two. 
 
 In the other direction, suppose that $\phi$ is a solution to $\X{J}$. Let $a$ be an agent whose bundle size is two and $\phi^{-1}(a)$ contains at least one \object from $\Co{I}_{\sf HighUtility}$, say $z$. Update the assignment $\phi$ as follows: $\phi^{-1}(a)=\{z\}$. Note that $\phi$ is still a solution to $\X{J}$.   Let $\Co{A}_{1} \subseteq \Co{A}$ be the set of agents such that for every agent $a\in \Co{A}_{1}$, $|\phi^{-1}(a)|=1$, i.e., the bundle size assigned to every agent in $\Co{A}_{1}$ is $1$. Clearly, $\phi^{-1}(\Co{A}_{1})\subseteq \Co{I}_{\sf HighUtility}$. Let ${\sf rem}= \Co{I}_{\sf HighUtility}\setminus \phi^{-1}(\Co{A}_{1})$, the set of unassigned ``high value'' \items. Suppose that ${\sf rem}\neq \emptyset$. 

Let $\Co{A}'' \subseteq \Co{A}\setminus \Co{A}_{1}$ be a set of size $\min\{|\Co{A}\setminus \Co{A}_{1}|,|{\sf rem}|\}$. Let $\Co{A}'' =\{a_1,\ldots,a_\ell\}$ and ${\sf rem}=\{z_1,\ldots,z_q\}$, where clearly $\ell \leq q$. 
Update the assignment $\phi$ as follows: for every $i\in [\ell]$, $\phi(z_i)=\{a_i\}$. Clearly, $\phi$ is still a solution of $\X{J}$. 
We note that there are only two cases: either $\Co{A}=\Co{A}_{1}\cup \Co{A}''$ or $\Co{\tilde{A}}=\Co{A}\setminus (\Co{A}_{1}\cup \Co{A}'')$ is non-empty. 

If $\Co{A}=\Co{A}_{1}\cup \Co{A}''$, then we have that the disjoint union of $\phi^{-1}(\Co{A}_{1}) \cup \phi^{-1}(\Co{A}'')\subseteq  \Co{I}_{\sf HighUtility}$. In other words,  $|\Co{I}_{\sf HighUtility}| \geq n$, and  
so there exists a solution in which every bundle is of size one and contains an element from  $\Co{I}_{\sf HighUtility}$. 

Otherwise, let $\Co{\tilde{A}}=\Co{A}\setminus (\Co{A}_{1}\cup \Co{A}'')$. Clearly, each of the \items in $\Co{I}_{\sf HighUtility}$ are assigned to agents in $\Co{A}_{1}\cup \Co{A}''$ and subsets of \items in $\Co{I}_{\sf LowUtility}$ are assigned to agents in $\Co{\tilde{A}}$. In other words, there exist $|\Co{I}_{\sf HighUtility}|$ bundles of size one and $n- |\Co{I}_{\sf HighUtility}|$ bundles of size two. Specifically for the latter, we know that each of the bundles is an independent set, they are pairwise disjoint, and the total utility within each bundle is at least $\eta$. Thus, the members of each bundle share an edge in the graph $\Co{\widehat{H}}$ and the bundles themselves form a matching in the graph. Thus, our algorithm that computes a maximum matching in $\Co{\widehat{H}}$ would find a matching of size at least $n- |\Co{I}_{\sf HighUtility}|$. Hence, given the construction of the assignment from such a matching, we can conclude that our algorithm would return an assignment with the desired properties.  

\subparagraph{For the complete variant.} Note that the value of $s=2$ as argued above. Thus, \ccffa and \sbccffa variants are the same. 

Let $\X{J}=(\Co{A}, \I, \{\util_{i}\}_{i\in [n]}, \Co{H}, \eta)$  be a given instance of \ccffa, where $\Co{H}$ is a cluster graph with two cliques, and the utility functions are uniform. Let $C_1$ and $C_2$ be two cliques in $\Co{H}$. The definitions of $\Co{I}_{\sf HighUtility}$ and $\Co{I}_{\sf LowUtility}$ are same as defined earlier. We first guess the number of items $\ell_1=|\Co{I}_{\sf HighUtility} \cap C_1|$ and $\ell_2=|\Co{I}_{\sf HighUtility} \cap C_2|$ that are assigned as singletons in a solution. Next, we construct an auxiliary graph $\widehat{H}$ as follows: the vertex set of $\widehat{H}$ contains all the items in $\Co{I}_{\sf LowUtility}$, any $|\Co{I}_{\sf HighUtility}|-\ell_1$ items from the set $\Co{I}_{\sf HighUtility} \cap C_1$, and any $|\Co{I}_{\sf HighUtility}|-\ell_2$ items from the set $\Co{I}_{\sf HighUtility} \cap C_2$; for any two items $x\in C_1$ and $y\in C_2$, $xy$ is an edge if $\util(x)+\util(y)\geq \eta$. We find a perfect matching $M$ in $\widehat{H}$ in polynomial time. If $\ell_1+\ell_2+|M|=n$, we return ``yes''. If the algorithm does not return ``yes'' for choice of $\ell_1$ and $\ell_2$, then we return ``no''. 

\noindent{\bf Correctness:} Note that if we return ``yes'', then clearly, it's a yes-instance of the problem as we can assign pair of matched items to any $|M|$ agents, $\ell_1$ items from $(\Co{I}_{\sf HighUtility} \cap C_1)\setminus V(\widehat{H})$ and $\ell_1$ items from $(\Co{I}_{\sf HighUtility} \cap C_2)\setminus V(\widehat{H})$. Clearly, all the constraints are satisfied. For the other direction, suppose that $\Co{J}$ is a yes-instance of \ccffa, and let $\phi\colon {\Co{I}} \rightarrow {\Co{A}}$ be one of its solution. Let $\Co{A}' \subseteq \Co{A}$ be the set of agents that are assigned single item. Clearly, $\phi^{-1}({\Co{A}'})\subseteq \Co{I}_{\sf HighUtility}$. Let $\ell_1=|\phi^{-1}({\Co{A}'})\cap C_1|$ and $\ell_2=|\phi^{-1}({\Co{A}'})\cap C_1|$. Note that all the items in $\Co{I}_{\sf HighUtility}$ are identical in the sense that they can be assigned to any agent as singleton and achieve the target value. Clearly, we guessed these $\ell_1$ and $\ell_2$ as well. Note that every agent in $\Co{A}\setminus \Co{A}'$. Thus, the auxiliary graph $\widehat{H}$ constructed for these two choices of $\ell_1$ and $\ell_2$ has a perfect matching of size $|\Co{A}\setminus \Co{A}'|$, hence we return ``yes''. 

Next, we prove the hardness result.
\begin{enumerate}[resume]
\item {\sf NP}-hardness
\end{enumerate}
The proof is similar to the {\sf NP}-hardness proof  in Theorem~\ref{thm:np-hardness-results} for the size bounded variant. The only difference is the construction of the conflict graph \calH. Here, we add edges between all the \items corresponding to every element in $X$ and $D$. Furthermore, we add edges between all the \items corresponding to every element in $Y$. Thus, it is a cluster graph with two cliques. The rest of the construction and the proof is the same.
\end{proof}

\section{Proof of \Cref{thm:kernel-partial-complete}}\label{sec:kernel-partial-complete}
\partialcompletekernel*
\begin{proof}~
\begin{enumerate}
\item Graphs of bounded neighborhood diversity.
\end{enumerate}
First, we focus initially on the \completefairdiv problem and  present a polynomial kernel. Subsequently, we make slight adjustments to our marking procedure to achieve the same outcome for the \partialfairdiv problem. To begin, we reaffirm the definition of {\em neighborhood diversity} within our specific context.
Given a graph, we say two vertices $u,v$ are of the same type if $N(v)\setminus\{u\}=N(u)\setminus\{v\}$. Consider $\tau$ as the number of types of vertices (also known as the neighborhood diversity) in the conflict graph $G$, with vertex groups $V_1\uplus V_2\uplus\ldots\uplus V_\tau$. Leveraging the intrinsic characteristic of neighborhood diversity, each group $V_i$, where $i\in[\tau]$, is either a clique or constitutes an independent set. Moreover, for any pair $i,j\in[\tau]$ where $i\neq j$, either every vertex in $V_i$ is adjacent to each vertex in $V_j$ or none of them are adjacent. Observe that if there exists a clique of size at least $n+1$ in the conflict graph, then the provided instance of \completefairdiv is a \no-instance. Consequently, we proceed with the assumption that each $V_i$ that forms a clique has a size of at most $n$. Without loss of generality let the groups  $V_i$s where $i\in[\tau']$ constitutes independent sets and each of the remaining groups constitutes a clique. Let $V'=\bigcup _{i=\tau'+1}^{\tau} V_i$ be the set of vertices of the cliques. Building upon this, we implement the subsequent marking scheme that aids us in constraining the number of vertices within each group $V_i$ that constitutes an independent set (as opposed to a clique).

 \begin{center}
     \begin{enumerate}
         \item[--] For a group $V_i$ where $i\in[\tau']$, initialize $B_i(j)=\emptyset$, $\forall j\in[n]$. 
         \item[--] For an arbitrary vertex $v\in V_i$, if there exists $j\in[n]$, such that $\util_j(v)>0$ and $|B_i(j)|<\eta n$, include $v$ into $B_i(j)$.
          
         \item[--] Delete $v$, if $v$ is not included in any $B_i(j)$, where $i\in[\tau']$ and $j\in[n]$.
			     
     \end{enumerate}
 \end{center}

Following the marking procedure, the size of any bag  $B_i(j), i\in[\tau']$ and $j\in[n]$ is bounded by $\eta n$. Therefore, the number of remaining vertices in each group $V_i$ is at most $\eta n^2$.
\begin{claim}\label{claim:nbrdiver}
$(\calA, \calI, \{\util_i\}_{a_i\in \calA}, \Co{H}, \eta)$ is a \yes-instance if and only if $(\calA, \calI, \{\util_{i}\}_{i\in [n]}, \Co{H}[V^*], \eta)$ is a \yes-instance, where $V^*=V' \bigcup_{i\in[t],j\in[n]}B_i(j)$.
\end{claim}
\noindent
{\em Proof:} In the forward direction, let's consider a \yes-instance $(\Co{A}, \Co{I},\{\util_{i}\}_{i\in [n]}, \Co{H}, \eta)$ of \completefairdiv accompanied by a solution $\mathcal{S'}$. This solution encapsulates a \textit{minimal} set denoted as $\mathcal{S}=(S_1,\ldots, S_n)$, in which each $S_i$ is assigned to the agent $a_i$ and $\sum_{v\in S_i}\util_{i}(v)\geq \eta$. In our context, a set is classified as \textit{minimal} if it has no proper subset that is \textit{minimal}. 

Next, we construct an assignment completely contained in $V^*$ where each agent receives an independent set and has its utility requirement of $\eta$ satisfied. If the assignment $\mathcal{S}_{V^*}=(S_1\cap V^*,\ldots,S_n\cap V^*)$, where each $S_i\cap V^*$ is allocated to agent $a_i$, does not meet the utility requirements, then there exists an index $i\in [n]$ such that $\sum_{v\in S_i\cap V^*}\util_{i}(v)<\eta$. Independence is not a concern as $S_i\cap V^*$ induces an independent set. However, in such a scenario there exists a (deleted) vertex $v\in (S_i\cap B_j)\setminus V^*$, where $B_j=\Sc\cap V_j$. This implies that $\eta n$ many vertices from $B_j$ were already marked in the marking scheme (with positive utilities) for agent $a_i$. Among these vertices, at most $(n-1)\eta$ may be allocated to other agents within the minimal solution. Hence within $(V_j\cap V^*)\setminus \bigcup_{l\in[n], l\neq i} S_l$, there are at least $\eta$ vertices that have positive utility values for agent $a_i$ and are not assigned to any other agents in $\Sc$. We construct an alternate solution by assigning all these vertices to agent $a_i$, thereby fulfilling its utility requirement of $\eta$. Employing this procedure exhaustively leads to an assignment with all assigned vertices in $V^*$ where each agent receives an independent set and has its utility requirement of $\eta$ satisfied. This implies $(\Co{A}, \Co{I},\{\util_{i}\}_{i\in [n]}, \Co{H}[V^*], \eta)$ is indeed a \yes-instance.

The converse direction is straightforward to establish. Let $\mathcal{S}^*=(S_1,\ldots,S_n)$ represent a solution to $(\Co{A}, \Co{I},\{\util_{i}\}_{i\in [n]}, \Co{H}[V^*], \eta)$. Considering that there must exist an agent that is assigned a vertex from $V_j\cap V^*$, we can extend the solution $\mathcal{S}'$ to a solution for $(\Co{A}, \Co{I},\{\util_{i}\}_{i\in [n]}, \Co{H}, \eta)$ by assigning all the additional vertices from $V_j$ to the respective agent. This confirms the correctness of the reduction. \hfill $\diamond$

Thus, we obtain a polynomial kernel when the conflict graph has a bounded neighborhood diversity.
It's important to observe that for \partialfairdiv we can not bound the size of the maximum clique by $n$ using the same argument. However, we can attain a comparable bound as follows. In each clique, for every agent, we retain/mark $n^2$ vertices with the highest utility values corresponding to the agent. And following this marking scheme we discard the remaining vertices. The correctness of this marking procedure is evident from the fact that any agent can pick at most one vertex from such a clique and since all vertices are of the same type, it prefers/chooses to pick one with utility value as large as possible. This marking procedure (for groups that are cliques), combined with the marking procedure (for groups that induce independent sets) for the \completefairdiv problem, produces a kernel of the same asympotic size ($\tau \eta n^2$) for the \partialfairdiv problem.



\begin{enumerate}[resume]
\item Graphs of bounded degree ($d<n$).
\end{enumerate}
Now, we provide a kernel for the \partialfairdiv where the maximum degree of the conflict graph is less than $n$. Subsequently,  we show that the same kernelization algorithm can be extended to the \completefairdiv problem. Towards that let $d$ be the maximum degree of the graph $G$. We proceed by describing the following marking procedure. 
  \begin{center}
     \begin{enumerate}
         \item[--] For each agent $a_i$, initialize $B_i=\emptyset$, $\forall$ $i\in[n]$. 
         \item[--] For an arbitrary vertex $v\in V$, if there exists $i\in[n]$, such that $\util_i(v)>0$ and $|B_i|<d\eta n$, include $v$ into $B_i$.

         \item[--] Delete $v$, if $v$ is not included in any $B_i$ where $i\in[n]$.
			     
     \end{enumerate}
 \end{center}

We reiterate that the input instance of \partialfairdiv problem is denoted by $(\calA, \calI, \{\util_{i}\}_{i\in [n]}, \Co{H}, \eta)$.
Let  $V^*=\bigcup_{i\in[n]}B_i$. Below, we claim $V^*$ to be our desired kernel.

\begin{claim}\label{lemma:kernel_degree}
$(\calA, \calI, \{\util_{i}\}_{i\in [n]}, \Co{H}, \eta)$ is a \yes-instance if and only if $(\calA, \calI, \{\util_{i}\}_{i\in [n]}, \Co{H}[V^*], \eta)$ is a \yes-instance.
\end{claim}
 {\em Proof:}   In the forward direction let $(\Co{A}, \Co{I},\{\util_{i}\}_{i\in [n]}, \Co{H}, \eta)$ be a \yes-instance. The key observation used here is, a graph with $m$ vertices and bounded degree $d$, has an independent set of size at least $\lfloor \frac{m}{d}\rfloor$. Since for each agent $a_i$, we mark $d\eta n$ many vertices (each with a positive utility for $a_i$), there is an independent set of size at least $\eta n$ among such marked vertices. Moreover, in any minimal solution, one agent is assigned at most $\eta$ many vertices (from such an independent set). Thus from such an independent set, at most $(n-1)\eta$ vertices are assigned to agents other than $a_i$. Hence we may assign at least $\eta$ many vertices (who are also independent) to agent $a_i$ with positive utilities, satisfying the utility requirement. Note that the aforementioned arguments hold true for all the other agents as well. Therefore $(\Co{A}, \Co{I},\{\util_{i}\}_{i\in [n]}, \Co{H}[V^*], \eta)$ is indeed a \yes-instance.

The reverse direction follows from the fact that $V^*\subseteq V$. Any satisfying assignment for \partialfairdiv on a subgraph remains a satisfying assignment on the original graph. 
\hfill $\diamond$


\noindent
When we focus on the \completefairdiv problem, it's important to note that we can bound the size of a maximum clique by $n$ for a \yes-instance. And, the implication of $d<n$ from it is used to design a kernel for the \completefairdiv by making use of the marking procedure described above. Let $\mathcal{S}=(S_1,\ldots,S_n)$ be any solution for the instance $(\calA, \calI, \{\util_{a}\}_{a\in \calA}, \Co{H}, \eta)$ of \partialfairdiv. Consider any vertex $v$ that was either deleted during the marking procedure ($v \in V\setminus V^*$) or belongs to $V^*\setminus S$. Since $\deg (v)< n$, there must exist a set $S_i$ that does not contain any neighbors of $v$. Hence we can assign $v$ to agent $a_i$, thus expanding the earlier solution by including $v$ into it, i.e., $\mathcal{S}'=(S_1,\ldots,S_i\cup \{v\}, \ldots, S_n)$ for the instance $(G[V^* \cup \{v\}],n,\eta,\widehat{s}+1,\{\util_{i}\}_{i\in [n]})$. An exhaustive application of this  expansion procedure results in a solution for the instance $(G,n,\eta,m,\{\util_{i}\}_{i\in [n]})$ of \completefairdiv. Hence we have a kernel of size $d\eta n^2$ for both the version  of the problem.



\begin{enumerate}[resume]
\item Graphs of bounded treewidth.
\end{enumerate}
Due to~\cite{Jansan}, there is no polynomial kernel for the \mwis problem when paramaterized by vertex cover ({\sf vc}). \partialfairdiv problem is essentially the \mwis when $k=1$, Hence \partialfairdiv problem does not admit a polynomial kernel when parameterized by $\vc+n$ and also for the parameter ${\sf tw}+n$. Moreover, note that \kcol has no polynomial kernel parameterized by the combined parameter {\sf tw}+$n$, as there are straight-forward \textit{AND-cross-compositions} from the problem to itself ~\cite{cygan2015parameterized}.
Consequently, the same conclusion extends to the \completefairdiv problem even for parameter $\tw+\eta(=0)+n$
\end{proof}

\section{Proof of \Cref{thm:kernel-contrast}}
\contrastkernel*
\begin{proof}
This theorem discusses the contrast of results between the partial version of the problem with the complete version. We first consider \pcffa and provide kernels for respective parameters and then show the impossibility results for \ccffa.
\begin{enumerate}
\item Graphs of maximum degree $d>n$.
\end{enumerate}
Observer that in proof of \Cref{thm:kernel-partial-complete}, we have not used the restriction on the maximum degree to get the kernel for the \partialfairdiv. Thus the same result holds true in this case as well. However, \Cref{lemma:chromatic-clique-degree hard} refutes any possibility of getting a kernel for the \completefairdiv problem.  
\begin{enumerate}[resume]
\item Graphs of bounded chromatic number ($\chi$) and excluding large cliques ($K_r$).
\end{enumerate}
We employ a marking scheme similar to that used for graphs with bounded degree (see \Cref{sec:kernel-partial-complete}). Let $\Co{H}_\chi$ denote the collection of graphs with a chromatic number not exceeding $\chi$. The crucial insight that any set of $\chi \cdot b$ vertices inherently encompasses an independent set with a size of at least $b$. Building upon this pivotal observation, we establish the subsequent marking procedure. 
  \begin{center}
     \begin{enumerate}
         \item[--] For each agent $a_i$, initialize $B_i=\emptyset$, $\forall$ $i\in[n]$. 
         \item[--] For an arbitrary vertex $v\in V$, if there exists $i\in[n]$, such that $\util_i(v)>0$ and $|B_i|<\chi\eta n$, include $v$ into $B_i$.

         \item[--] Delete $v$, if $v$ is not included in any $B_i$ where $i\in[n]$.
			     
     \end{enumerate}
 \end{center}
The safeness of the marking procedure follows in a manner analogous to that of Lemma~\ref{lemma:kernel_degree}. However, a partial assignment may not be extendable to a complete one in this context.

    When the conflict excludes large clique $K_r$, we use the {\em Ramsey's Theorem} as a pivotal tool to design the desired kernel.
    \begin{proposition}
       Any graph $G$ of size $R(r,r_i)$ contains either a clique of size $r$ i,e., $K_r$ or an independent set of size $r_i$ where $R(r,r_i)$ is $\Oh((r+r_i)^r)$. 
    \end{proposition}
 As before, our goal is to select a sufficient number of vertices with positive utilities for each agent, ensuring that there's a subset of these vertices that forms an independent set of the size at least $\eta n$. 
Let $\Co{H}_r$ be the family of graphs with no $K_{r}$. For each agent, we mark at most $R(r,\eta n)$ vertices with positive utilities (corresponding to these agents) in the marking scheme. But $G \in \Co{H}_r$, of size $R(r,\eta n)$ always contains an independent set of size least $\eta n$. Formally, the marking scheme is as follows.

  \begin{center}
     \begin{enumerate}
         \item[--] For each agent $a_i$, initialize $B_i=\emptyset$, $\forall$ $i\in[n]$. 
         \item[--] For an arbitrary vertex $v\in V$, if there exists $i\in[n]$, such that $\util_i(v)>0$ and $|B_i|<R(r,\eta n)$, include $v$ into $B_i$.

         \item[--] Delete $v$, if $v$ is not included in any $B_i$ where $i\in[n]$.		     
     \end{enumerate}
 \end{center}
\noindent 
The safeness of the marking procedure can be reasoned in a manner similar to the explanation provided in the proof of \Cref{lemma:kernel_degree}.
\subparagraph{{\sf NP}-hardness.}
 Now, we consider the \ccffa problem and provide a reduction from $3$-{\sc Coloring} on 4-regular planar graph to \ccffa. Note that finding the exact chromatic number of a 4-regular planar graph (which is $4$-colorable and $K_5$-free) is \nph \cite{DAILEY1980289}. Given an instance of 3-{\sc Coloring} on planar graphs, we can reduce it to an equivalent instance $(\calA, \calI, \{\util_{i}\}_{i\in [n]}, \Co{H}, \eta)$ of \completefairdiv, where $|\calA|=3, \calI= V(G), \Co{H}=G$  $\eta=1$, and $\util_i(x)=1$,  $\forall a_i\in \calA$, $\forall x\in V(G)$. Since the planar graph $\Co{H}$ excludes $K_5$, has a bounded chromatic number and a bounded degree, we have the following observation.

\begin{observation}
    \label{lemma:chromatic-clique-degree hard}
    \completefairdiv is \nph even when $n+\eta+\chi +r$ or $n+\eta+d$ is a constant.
\end{observation}

It is easy to see that \Cref{lemma:chromatic-clique-degree hard} immediately eliminates the possibility of getting a kernel parameterized by $n+\eta +\chi + r$.
\end{proof}

\section{Conclusion}

In this article, we studied conflict-free fair allocation problem under the paradigm of parameterized complexity with respect to several natural input parameters. We hope that this will lead to a new set of results for the problem. 

One natural direction of research is to consider various other fairness notions known in the literature, such as envy-freeness, proportional fair-share, min-max fair-share, etc., under the conflict constraint. Some of these fairness notions have already been studied under this constraint, but not from the viewpoint of parameterized complexity.

\bibliography{references}

\begin{thebibliography}{10}

\bibitem{ahmadian2021four}
M.~M. Ahmadian, M.~Khatami, A.~Salehipour, and T.~C.~E. Cheng.
\newblock Four decades of research on the open-shop scheduling problem to
  minimize the makespan.
\newblock {\em European Journal of Operational Research}, 295(2):399--426,
  2021.

\bibitem{alon1995color}
N.~Alon, R.~Yuster, and U.~Zwick.
\newblock Color-coding.
\newblock {\em Journal of the ACM (JACM)}, 42(4):844--856, 1995.

\bibitem{DBLP:conf/atal/BarmanV21}
S.~Barman and P.~Verma.
\newblock Existence and computation of maximin fair allocations under
  matroid-rank valuations.
\newblock In {\em {AAMAS} '21}, pages 169--177, 2021.

\bibitem{bezakova2005allocating}
I.~Bez{\'a}kov{\'a} and V.~Dani.
\newblock Allocating indivisible goods.
\newblock {\em ACM SIGecom Exchanges}, 5(3):11--18, 2005.

\bibitem{Biswas_FORC_2023}
Arpita Biswas, Yiduo Ke, Samir Khuller, and Quanquan~C. Liu.
\newblock An algorithmic approach to address course enrollment challenges.
\newblock In {\em 4th Symposium on Foundations of Responsible Computing, {FORC}
  2023,}, volume 256 of {\em LIPIcs}, pages 8:1--8:23, 2023.

\bibitem{cheng2022restricted}
S.~Cheng and Y.~Mao.
\newblock Restricted max-min allocation: Integrality gap and approximation
  algorithm.
\newblock {\em Algorithmica}, pages 1--40, 2022.

\bibitem{DBLP:conf/iwoca/ChiarelliKMPPS20}
N.~Chiarelli, M.~Krnc, M.~Milanic, U.~Pferschy, N.~Pivac, and J.~Schauer.
\newblock Fair packing of independent sets.
\newblock In {\em {IWOCA} 2020}, volume 12126 of {\em Lecture Notes in Computer
  Science}, pages 154--165, 2020.

\bibitem{ParamAlgorithms15b}
M.~Cygan, F.~Fomin, L.~Kowalik, D.~Lokshtanov, D.~Marx, M.~Pilipczuk,
  M.~Pilipczuk, and S.~Saurabh.
\newblock {\em Parameterized Algorithms}.
\newblock Springer, 2015.

\bibitem{DBLP:journals/tcs/CyganP10}
M.~Cygan and M.~Pilipczuk.
\newblock Exact and approximate bandwidth.
\newblock {\em Theor. Comput. Sci.}, 411(40-42):3701--3713, 2010.

\bibitem{cygan2015parameterized}
Marek Cygan, Fedor~V Fomin, {\L}ukasz Kowalik, Daniel Lokshtanov, D{\'a}niel
  Marx, Marcin Pilipczuk, Micha{\l} Pilipczuk, and Saket Saurabh.
\newblock {\em Parameterized algorithms}, volume~5.
\newblock Springer, 2015.

\bibitem{Dabrowski12}
K.~Dabrowski, V.~V. Lozin, H.~M{\"u}ller, and D.~Rautenbach.
\newblock Parameterized complexity of the weighted independent set problem
  beyond graphs of bounded clique number.
\newblock {\em J. Discrete Algorithms}, 14:207--213, 2012.

\bibitem{DAILEY1980289}
David~P. Dailey.
\newblock Uniqueness of colorability and colorability of planar 4-regular
  graphs are np-complete.
\newblock {\em Discrete Mathematics}, 30(3):289--293, 1980.
\newblock URL:
  \url{https://www.sciencedirect.com/science/article/pii/0012365X80902368},
  \href {https://doi.org/https://doi.org/10.1016/0012-365X(80)90236-8}
  {\path{doi:https://doi.org/10.1016/0012-365X(80)90236-8}}.

\bibitem{deuermeyer1982scheduling}
B.~L. Deuermeyer, D.~K. Friesen, and M.~A. Langston.
\newblock Scheduling to maximize the minimum processor finish time in a
  multiprocessor system.
\newblock {\em SIAM Journal on Algebraic Discrete Methods}, 3(2):190--196,
  1982.

\bibitem{DBLP:conf/atal/EbadianP022}
S.~Ebadian, D.~Peters, and N.~Shah.
\newblock How to fairly allocate easy and difficult chores.
\newblock In {\em {AAMAS} 2022}, pages 372--380. International Foundation for
  Autonomous Agents and Multiagent Systems {(IFAAMAS)}, 2022.

\bibitem{fomin2014efficient}
F.~V. Fomin, D.~Lokshtanov, and S.~Saurabh.
\newblock Efficient computation of representative sets with applications in
  parameterized and exact algorithms.
\newblock In {\em {SODA}}, pages 142--151, 2014.

\bibitem{DBLP:books/fm/GareyJ79}
M.~R. Garey and David~S. Johnson.
\newblock {\em Computers and Intractability: {A} Guide to the Theory of
  NP-Completeness}.
\newblock W. H. Freeman, 1979.

\bibitem{golumbic2004algorithmic}
Martin~Charles Golumbic.
\newblock {\em Algorithmic graph theory and perfect graphs}.
\newblock 2004.

\bibitem{grzesik2022polynomial}
A.~Grzesik, T.~Klimo\v{s}ov\'{a}, M.~Pilipczuk, and Micha\l{} Pilipczuk.
\newblock Polynomial-time algorithm\&nbsp;for maximum weight independent set on
  p6-free graphs.
\newblock {\em ACM Trans. Algorithms}, 18(1), 2022.
\newblock \href {https://doi.org/10.1145/3414473} {\path{doi:10.1145/3414473}}.

\bibitem{DBLP:journals/corr/abs-2309-04995}
Sushmita Gupta, Pallavi Jain, and Saket Saurabh.
\newblock How to assign volunteers to tasks compatibly ? {A} graph theoretic
  and parameterized approach.
\newblock {\em CoRR}, abs/2309.04995, 2023.
\newblock URL: \url{https://doi.org/10.48550/arXiv.2309.04995}, \href
  {http://arxiv.org/abs/2309.04995} {\path{arXiv:2309.04995}}, \href
  {https://doi.org/10.48550/ARXIV.2309.04995}
  {\path{doi:10.48550/ARXIV.2309.04995}}.

\bibitem{DBLP:journals/siamcomp/HopcroftK73}
John~E. Hopcroft and Richard~M. Karp.
\newblock An n\({}^{\mbox{5/2}}\) algorithm for maximum matchings in bipartite
  graphs.
\newblock {\em {SIAM} J. Comput.}, 2(4):225--231, 1973.
\newblock \href {https://doi.org/10.1137/0202019} {\path{doi:10.1137/0202019}}.

\bibitem{hummel2022fair}
H.~Hummel and M.~L. Hetland.
\newblock Fair allocation of conflicting items.
\newblock {\em Autonomous Agents and Multi-Agent Systems}, 36(1):1--33, 2022.

\bibitem{Jansan}
Bart Jansen and Hans Bodlaender.
\newblock Vertex cover kernelization revisited: Upper and lower bounds for a
  refined parameter.
\newblock volume~53, pages 177--188, 03 2011.
\newblock \href {https://doi.org/10.1007/s00224-012-9393-4}
  {\path{doi:10.1007/s00224-012-9393-4}}.

\bibitem{keil2017algorithm}
J.~M. Keil, J.~SB Mitchell, D.~Pradhan, and M.~Vatshelle.
\newblock An algorithm for the maximum weight independent set problem on
  outerstring graphs.
\newblock {\em Computational Geometry}, 60:19--25, 2017.

\bibitem{kurokawa2018fair}
D.~Kurokawa, A.~D. Procaccia, and J.~Wang.
\newblock Fair enough: Guaranteeing approximate maximin shares.
\newblock {\em Journal of the ACM (JACM)}, 65(2):1--27, 2018.

\bibitem{lozin2008polynomial}
Vadim~V. Lozin and Martin Milani{\v c}.
\newblock A polynomial algorithm to find an independent set of maximum weight
  in a fork-free graph.
\newblock {\em Journal of Discrete Algorithms}, 6(4):595--604, 2008.
\newblock \href {https://doi.org/https://doi.org/10.1016/j.jda.2008.04.001}
  {\path{doi:https://doi.org/10.1016/j.jda.2008.04.001}}.

\bibitem{moenck1976practical}
R.~T. Moenck.
\newblock Practical fast polynomial multiplication.
\newblock In {\em Proceedings of {SYMSAC}'76}, pages 136--148, 1976.

\bibitem{naor1995splitters}
M.~Naor, L.~J. Schulman, and A.~Srinivasan.
\newblock Splitters and near-optimal derandomization.
\newblock In {\em Proceedings of IEEE 36th Annual Foundations of Computer
  Science}, pages 182--191, 1995.

\bibitem{DBLP:journals/sigecom/Suksompong21}
W.~Suksompong.
\newblock Constraints in fair division.
\newblock {\em SIGecom Exch.}, 19(2):46--61, 2021.

\bibitem{woeginger1997polynomial}
G.~J. Woeginger.
\newblock A polynomial-time approximation scheme for maximizing the minimum
  machine completion time.
\newblock {\em Operations Research Letters}, 20(4):149--154, 1997.

\end{thebibliography}

\end{document}